\documentclass[aps,pra,letterpaper,twocolumn,tightenlines,superscriptaddress,notitlepage,longbibliography,nofootinbib]{revtex4-2}
\usepackage{qcircuit}
\usepackage[dvips]{graphicx}
\usepackage{amsmath,amssymb,amsthm,mathrsfs,amsfonts,dsfont}
\usepackage{mathtools}
\usepackage{subfigure, epsfig}
\usepackage{braket}
\usepackage{bm}
\usepackage{enumerate}
\usepackage{xcolor}
\usepackage{comment}
\usepackage{scalefnt}
\usepackage[colorlinks = true, 
citecolor = blue,
linkcolor = blue,
urlcolor = blue]{hyperref}

\usepackage{pagecolor}
\usepackage[T1]{fontenc}
\usepackage{lmodern}
\usepackage[utf8]{inputenc}
\usepackage{microtype}

\usepackage{tikzscale}
\usepackage{pgfplots}
\pgfplotsset{compat=newest}
\usetikzlibrary{plotmarks}
\usetikzlibrary{arrows.meta}
\usetikzlibrary{external}
\usepgfplotslibrary{patchplots}
\usepackage{grffile}
\usepackage{enumitem}

\DeclareMathOperator{\Tr}{Tr}

\newcommand{\coleq}{\mathrel{\mathop:}\nobreak\mkern-1.2mu=}
\newcommand{\eqcol}{\mkern-1.2mu=\mathrel{\mathop:}\nobreak}

\newcommand{\mc}{\mathcal}

\newcommand{\mr}{\mathrm}
\newcommand{\mbb}{\mathbb}

\newcommand{\expval}[1]{{\langle #1 \rangle}}
\newcommand{\ketbra}[2]{{\vert #1 \rangle \langle #2 \vert}}
\newcommand{\lket}[1]{\vert #1 \rangle\!\rangle}
\newcommand{\lbra}[1]{\langle\!\langle #1 \vert}
\newcommand{\lbraket}[2]{\langle\!\langle #1 \vert #2 \rangle\!\rangle}
\newcommand{\lketbra}[2]{\vert #1 \rangle\!\rangle\langle\!\langle #2 \vert}
\newcommand{\ptm}{{\mathrm{PTM}}}
\newcommand{\E}{\mathop{\mbb{E}}}

\makeatletter
\newcommand{\pushright}[1]{\ifmeasuring@#1\else\omit\hfill$\displaystyle#1$\fi\ignorespaces}
\makeatother

\newtheorem{theorem}{Theorem}
\newtheorem{proposition}[theorem]{Proposition}

\newtheorem{definition}{Definition}

\definecolor{applegreen}{rgb}{0.55, 0.71, 0.0}

\newcommand{\comments}[1]{}
\usepackage{algorithm}  
\usepackage[noend]{algpseudocode}  
\newcommand{\algorithmfootnote}[2][\footnotesize]{%
  \let\old@algocf@finish\@algocf@finish%
  \def\@algocf@finish{\old@algocf@finish%
    \leavevmode\rlap{\begin{minipage}{\linewidth}
    #1#2
    \end{minipage}}%
  }%
}
\usepackage{xparse}
\NewDocumentCommand{\LeftComment}{s m}{%
  \Statex \IfBooleanF{#1}{\hspace*{\ALG@thistlm}}\(\triangleright\) #2}
\algnewcommand{\LineComment}[1]{\Statex // #1}

\begin{document}

\title{Tight bounds on Pauli channel learning without entanglement}
\author{Senrui Chen}
\email{csenrui@uchicago.edu}
\affiliation{Pritzker School of Molecular Engineering, The University of Chicago, Chicago, Illinois 60637, USA}
\author{Changhun Oh}
\affiliation{Pritzker School of Molecular Engineering, The University of Chicago, Chicago, Illinois 60637, USA}
\affiliation{Department of Physics, Korea Advanced Institute of Science and Technology, Daejeon 34141, Korea}
\author{Sisi Zhou}
\affiliation{Institute for Quantum Information and Matter,
California Institute of Technology, Pasadena, CA 91125, USA}
\affiliation{Perimeter Institute for Theoretical Physics, Waterloo, Ontario N2L 2Y5, Canada}
\author{Hsin-Yuan Huang}
\affiliation{Institute for Quantum Information and Matter,
California Institute of Technology, Pasadena, CA 91125, USA}
\affiliation{Center for Theoretical Physics,
Massachusetts Institute of Technology, Cambridge, MA 02139, USA}
\author{Liang Jiang}
\email{liang.jiang@uchicago.edu}
\affiliation{Pritzker School of Molecular Engineering, The University of Chicago, Chicago, Illinois 60637, USA}

\date{\today}

\begin{abstract}
    Quantum entanglement is a crucial resource for learning properties from nature, but a precise characterization of its advantage can be challenging.
    In this work, we consider learning algorithms without entanglement to be those that only utilize states, measurements, and operations that are separable between the main system of interest and an ancillary system.
    Interestingly, we show that these algorithms are equivalent to those that apply quantum circuits on the main system interleaved with mid-circuit measurements and classical feedforward.
    Within this setting, we prove a tight lower bound for Pauli channel learning without entanglement that closes the gap between the best-known upper and lower bound.
    In particular, we show that $\Theta(2^n\varepsilon^{-2})$ rounds of measurements are required to estimate each eigenvalue of an $n$-qubit Pauli channel to $\varepsilon$ error with high probability when learning without entanglement.
    In contrast, a learning algorithm with entanglement only needs  $\Theta(\varepsilon^{-2})$ copies of the Pauli channel.
    The tight lower bound strengthens the foundation for an experimental demonstration of entanglement-enhanced advantages for Pauli noise characterization.
\end{abstract}

\maketitle

Entanglement lies at the heart of quantum information science and technology, providing significant advantages over classical information processing in computation~\cite{nielsen2002quantum}, communication~\cite{gisin2007quantum, kimble2008quantum}, metrology~\cite{giovannetti2006quantum, giovannetti2011advances, polino2020photonic}, and many other aspects.
A recent line of research uses information-theoretic tools to obtain rigorous and exponential quantum advantages in learning~\cite{huang2022quantum,chen2022exponential,bubeck2020entanglement,huang2021information,chen2022tight,aharonov2022quantum,caro2022learning,chen2022quantum,chen2023complexity}. It is shown, both theoretically and experimentally, that quantum resources %
can bring significant speed-up for learning certain properties from the nature, \textit{e.g.}, learning expectation values of many observables for a quantum state~\cite{huang2022quantum}.
However, the connection between these quantum advantages with specific quantum resources, \textit{e.g.}, quantum entanglement, is far from clear. 
This problem is prominent for learning properties from quantum channels, where there are many different ways of defining a ``quantum-enhanced'' experiments~\cite{aharonov2022quantum,huang2021information,chen2022exponential,chen2022quantum}, depending on whether one allows ancillary systems, concatenation of channels, mid-circuit controls, \textit{etc}. 
A scenario that precisely captures the role of quantum entanglement in learning is under exploration.

Apart from studying learning schemes from a quantum resource-theoretic perspective~\cite{chitambar2019quantum}, one can also take an operational approach. Specifically, a class of quantum operations known as mid-circuit measurement and classical feedforward have drawn increasing attention recently. While they are important building blocks of fault-tolerant quantum computation~\cite{shor1996fault,gottesman1998theory} and have found applications in recent experiments~\cite{iqbal2023topological,singh2023mid}, a framework that systematic study the effectiveness of mid-circuit measurements and classical feedforward in learning is yet to be established.

Turning to concrete learning tasks, a class of quantum channels that has drawn particular interest is the Pauli channel, which is defined to be a stochastic mixture of (multi-qubit) Pauli operations~\cite{nielsen2002quantum}. 
The Pauli channel is not only a basic model in quantum information theory, but also plays a crucial role in characterizing noisy quantum systems, with applications in quantum benchmarking~\cite{erhard2019characterizing,carignan2023error,harper2020efficient}, quantum noise mitigation~\cite{van2023probabilistic,Ferracin2022Efficiently,kim2023evidence}, quantum error correction~\cite{tuckett2018ultrahigh}, etc.
Techniques such as randomized compiling can engineer general quantum noise into Pauli channel under realistic assumptions~\cite{wallman2016noise,hashim2020randomized}.
A prerequisite for many of these applications is to learn an unknown Pauli channel.
Therefore, it is natural to study the protocols and limitations of Pauli channel learning. There have been several recent works exploring this direction~\cite{chen2022quantum,flammia2020efficient,flammia2021pauli,fawzi2023lower}. Specifically, Ref.~\cite{chen2022quantum} studies the sample complexity of Pauli channel estimation using information-theoretic methods, and shows an exponential separation between using and not using ancilla for learning every eigenvalue of an $n$-qubit Pauli channel to $\pm\varepsilon$ precision. 
However, as shown by the current work, the ancilla-free lower bound given there, $\Omega(2^{n/3})$, was not tight. While an ideal ancilla-assisted protocol has sample complexity $\Theta(1/\varepsilon^2)$\footnote{
There is an additional factor of $n$ in the upper bound of \cite{chen2022quantum} because the task there is to estimate all $\lambda_a$ simultaneously with high probability. Here we only require estimating each $\lambda_a$ individually with high probability.
One can derive the former using the latter via the union bound and paying a factor of $n$ in sample complexity.
}, real-world imperfections such as state preparation and measurement (SPAM) noise can introduce a weak exponential sampling overhead (see Fig.~\ref{fig:performance}). Therefore, to establish an advantage with moderate system sizes and realistic levels of imperfection, it is highly desirable to tighten the ancilla-free lower bound.

\begin{figure}[t]
    \centering
    \includegraphics[width=\columnwidth]{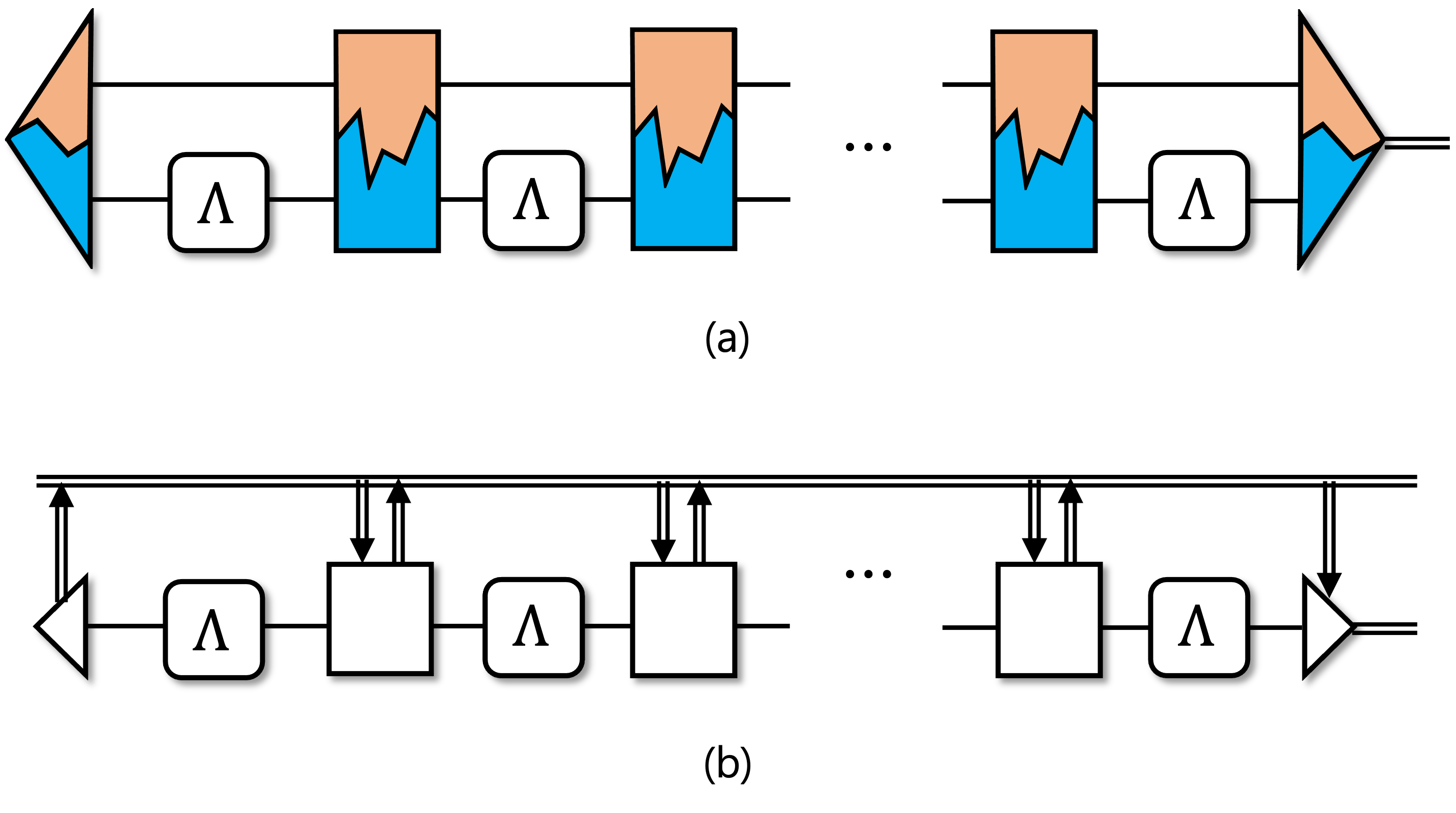}
    \caption{(a) Separable schemes. The two different colors indicate the operations are separable. 
    (b) Classical-memory-assisted schemes. The double line represents classical registers. The square box represents adaptively-chosen (arrows coming from the classical registers) quantum instruments (outcomes sent to the classical registers). We will show the two schemes are equivalent in terms of sample complexity, so we call both entanglement-free schemes.}
    \label{fig:schematics}
\end{figure}

In this work, we introduce a class of learning schemes that do not exploit entanglement between the main system and the ancillary system, from a resource-theoretic perspective. We also introduce a class of schemes that describe quantum circuits assisted with mid-circuit measurement and classical feedforward, from an operational perspective. (See Fig.~\ref{fig:schematics}.) Perhaps surprisingly, we show the two schemes are equivalent in terms of sample complexity for any learning tasks. 
This provides a new operational interpretation for quantum entanglement as a resource~\cite{bennett1996concentrating,bennett1996mixed,vedral1997quantifying}.
We then show that information-theoretic methods can be used to prove sample complexity lower bound in the above scenario. For the task of Pauli channel learning mentioned above, we obtain a tight lower bound of $\Omega(2^n/\varepsilon^2)$, closing the cubic gap with the known upper bound, and providing a tight exponential separation with the entanglement-assisted bound of $\Theta(1/\varepsilon^2)$~\cite{chen2022quantum}. 
We show that, this separation persists even if the entanglement-assisted scheme suffers from a reasonable amount of realistic noise.
Finally, we show how our results impose an limit on the efficiency of characterizing gate-dependent Pauli noise channels.

\medskip
\section{Setup}
Consider the task of learning properties of an $n$-qubit quantum channel $\Lambda$ from certain family by querying multiple copies of the channel.
We start by defining entanglement-free learning schemes. 
We introduce the class of \emph{separable schemes}, where a main system $\mc H_S$ and an ancillary system $\mc H_A$ are given. The main system is where $\Lambda$ acts on and has a fixed dimension of $2^n$, while the ancillary system can be arbitrarily large. Now, a separable scheme allows interleaving copies of $\Lambda$ on $\mc H_S$ with any processing operations (including state preparation, measurement, and quantum channels) on $\mc H_S\otimes\mc H_A$, with the only restriction that all the processing operations are separable~\cite{cirac2001entangling}  across $\mc H_A$ and $\mc H_S$. A schematic for separable schemes is shown in Fig.~\ref{fig:schematics} (a). 
The separable operations are known to be the largest set of operations that do not generate entanglement from un-entangled states (even when acting only on a subspace)~\cite{cirac2001entangling}, and is thus a suitable model for entanglement-free strategies from a quantum-resource-theoretic~\cite{chitambar2019quantum} point of view. 
Furthermore, separable operations contain as a subset other physically-motivated classes of operations like LOCC (local operation and classical communications)~\cite{bennett1993teleporting,chitambar2014everything}, 
which is also a standard choice of free operation in the resource theory of entanglement~\cite{bennett1996concentrating,bennett1996mixed,vedral1997quantifying}. 
A lower bound on sample complexity for the former implies a lower bound for the latter.

Besides separable schemes, we introduce another operationally motivated class of schemes called \emph{classical-memory-assisted schemes}. Here, one can only access the main system $\mc H_S$ (with a fixed dimension of $2^n$) and arbitrarily many classical registers. The allowed operations are to interleave copies of $\Lambda$ with adaptively-chosen \emph{quantum instruments}, which are defined as quantum channels associated with outputs to classical registers. By ``adaptively'' we mean the quantum instruments can be chosen according to the classical registers. A schematic is given in Fig.~\ref{fig:schematics}~(b). 
One can think of such schemes as quantum circuits assisted by mid-circuit measurement and classical feedforward control.
We remark that the classical-memory-assisted schemes include the ancilla-free concatenating scheme introduced in~\cite{chen2022quantum} as a special case, which can describe most randomized benchmarking (RB) ~\cite{emerson2005scalable,knill2008randomized} type protocols. The lower bounds obtained in this work thus also holds for those protocols.

While the above two schemes are introduced with different motivations, perhaps surprisingly, they are equivalent in terms of sample complexity. We have the following result. 
\begin{proposition}\label{prop:equiv}
For any separable scheme $A$, there exists a classical-memory-assisted scheme $B$ that generates the same outcome distribution as $A$ for any underlying $\Lambda$ using the same number of copies. Vice versa.
\end{proposition}
\noindent The formal definitions of both schemes, rigorous statement of Proposition~\ref{prop:equiv}, and the proof are given in SM~Sec.~\ref{app:model}.
Thanks to Proposition~\ref{prop:equiv}, we see that both schemes capture the power of learning without entanglement, and be treated interchangeably when studying the sample complexity. In the remaining part of this paper, we will focus on the classical-memory-assisted scheme, which has a clearer operational meaning. Specifically, there can be two different notions of complexity: One is the sample complexity $N_\mr{samp}$, which is the number of copies of $\Lambda$; The other is the number of measurements $N_\mr{meas}$, which is the number of quantum instruments with non-trivial measurement.
Clearly, $N_\mr{samp}\ge N_\mr{meas}$, as one can concatenate multiple copies of $\Lambda$ and only make one measurement, just like in RB. The lower bound we derive later will hold for $N_\mr{meas}$.

\bigskip
\section{Bounds on Pauli channel learning}
Having set up the formalism, we now study the specific problem of Pauli channel learning. An $n$-qubit Pauli channel $\Lambda$ has the following two equivalent forms,
\begin{equation}
    \Lambda(\rho)\coleq\sum_{b\in{\sf P}^{n}}p_b P_b\rho P_b =   \frac{1}{2^n}\sum_{a\in{\sf P}^{n}}\lambda_aP_a\Tr[P_a\rho],
\end{equation}
where ${\sf P}^n=\{I,X,Y,Z\}^{\otimes n}$ is the $n$-qubit Pauli group (modulo phase), $\{p_b\}_b$ is the Pauli error rates, and $\{\lambda_a\}_a$ is the Pauli eigenvalues~\cite{flammia2020efficient}.
Note that Pauli eigenvalues are also known as Pauli fidelities, which have been useful in quantum benchmarking~\cite{erhard2019characterizing,hashim2020randomized,flammia2021averaged,carignan2023error}, quantum error mitigation~\cite{chen2020robust, van2023probabilistic, Ferracin2022Efficiently}, etc. %
The task we consider is to learn each of the Pauli eigenvalues $\lambda_a$ to additive precision $\varepsilon$ with high success probability.
More precisely, we have the following result.
\begin{theorem}\label{th:main}
    If there exists an entanglement-free scheme that, for any $n$-qubit Pauli channel $\Lambda$, outputs an estimator $\widehat\lambda_a$ such that 
    $|\widehat\lambda_a-\lambda_a|\le\varepsilon\le1/6$ with probability at least $2/3$ for any $a\in{\sf P}^n$, after making $N$ rounds of measurement, then $N=\Omega(2^{n}/\varepsilon^2)$.
\end{theorem}
This matches the known upper bound of $O(2^n/\varepsilon^2)$ based on minimal stabilizer covering~\cite{flammia2020efficient,chen2022quantum},
solving an open problem raised therein. 
Note that the task we consider is to estimate each $\lambda_a$ with $2/3$ success probability \emph{individually} rather than \emph{simultaneously}. For the latter task, there will be an additional factor of $n$ in the upper bound, but our lower bound still holds and is tight up to this logarithmic factor.
Combined with the bound of $\Theta(1/\varepsilon^2)$ using entanglement-assisted scheme~\cite{chen2022quantum}, this gives a tight exponential separation for learning with and without entanglement in the task of Pauli channel learning. Another noteworthy feature is that our lower bound has a moderate constant factor. For example, with $\varepsilon\le0.1$ and $n\ge5$ we have $N\ge 0.01 \times 2^n/\varepsilon^2$.

To highlight the experimental relevance of our result, in Fig.~\ref{fig:performance} we plot our lower bound of Theorem~\ref{th:main}, the best previously known ancilla-free lower bound from~\cite{chen2022quantum}, and the upper bound from an entanglement-assisted scheme studied in~\cite{chen2022quantum} with noisy Bell states preparation (see SM.~Sec.~\ref{app:numeric} for details).
Fig.~\ref{fig:performance} clearly indicates that our improved lower bound is crucial for demonstrating the entanglement-enabled advantages with 
a moderate number of qubits and fidelity. 
For example, with Bell pair fidelity $95\%$, the previous lower bound needs at least $85$ qubits to start seeing any separation, while our improved lower bound needs as few as $25$ qubits to obtain a factor of $10^5$ advantages in sample complexity; With Bell pair fidelity below $90\%$, only our improved lower bound is able to obtain any separation.

\begin{figure}[tp]
    \centering
    \includegraphics[width=\columnwidth]{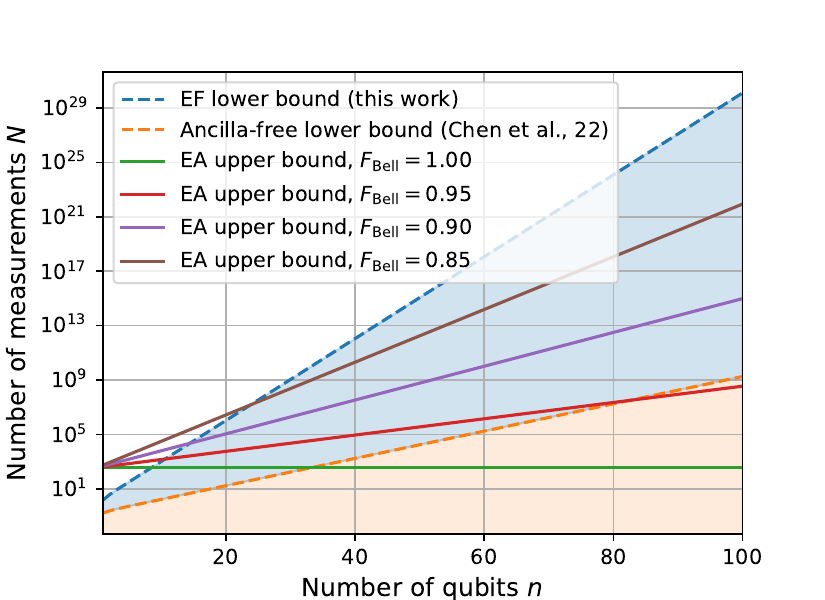}
    \caption{Sample complexity for Pauli channel learning. The task is to estimate any Pauli fidelity to $\varepsilon=0.1$ additive precision with at least $2/3$ success probability. The dash lines represent our entanglement-free (EF) lower bound of Thm.~\ref{th:main} and the ancilla-free lower bound from~\cite{chen2022quantum}. The solid lines represent the sample complexity upper bound calculated from an entanglement-assisted (EA) scheme with noisy Bell state and measurements. For simplicity, we assume the state preparation suffers from depolarizing noises, so that each noisy $2$-qubit Bell pair $\tilde{\rho}_\mr{Bell}$ has Fidelity $F_\mr{Bell}=\bra{\Psi_+}\tilde\rho_\mr{Bell}\ket{\Psi_+}$. The colored region indicates entanglement(ancilla)-enabled advantages.}
    \label{fig:performance}
\end{figure}

\begin{proof}[Proof Sketch of Theorem~\ref{th:main}.]
We extend the framework for proving exponential separable between learning with and without quantum memory~\cite{huang2022quantum,chen2022exponential}. 
The key idea is known as the Le Cam's two-point method~\cite{lecam1973convergence} that reduces learning to hypothesis testing. Specifically, we first construct two hypotheses of Pauli channels (or, mixture of Pauli channels) that are close to each other. 
By assumption, a learning scheme can distinguish the two hypotheses with good probability.
Consequently, the total variation distance (TVD) between the outcome probability distribution generated by the scheme under the two hypotheses needs to be at least constantly large.
Therefore, if we can upper bound the contributions to the TVD from each measurement to be exponentially small, we will obtain an exponential lower bound on the number of measurements.
However, the existing techniques for upper bounding TVD~ \cite{huang2022quantum,chen2022exponential} do not carry over when channel concatenation is allowed, let alone mid-circuit measurements.
Our technique to address this issue is to establish a recurrence relation on the mid-circuit states between each measurement step,
and to upper bound the growth of TVD via mathematical induction. 
The full proof is presented in SM.~Sec.~\ref{app:bound}.
\end{proof}
\medskip

\section{Bounds on learning identifiable parameters}
In practice, Pauli channels are often used to model noise affecting Clifford gates~\cite{wallman2016noise,hashim2020randomized,van2023probabilistic}. One issue for learning \emph{gate-dependent} noise channel is that, because of the existence of SPAM error, certain parameters of the noise channel might become non-identifiable (or, ``unlearnable''), meaning that they cannot be identified independently from the noisy SPAM~\cite{Blume-Kohout2013Robust,proctor2017randomized, nielsen2021gate, huang2022foundations,chen2023learnability}.
Specifically, for Clifford gate-dependent Pauli noise channel, a complete characterization of the learnable parameters is given in \cite{chen2023learnability}, which shows that some Pauli eigenvalues $\lambda_a$ cannot be identified SPAM-independently, but the geometric mean of certain set of eigenvalues, $\bigl(\prod_{a\in S}\lambda_a\bigr)^{1/|S|}$, can be. 
This is consistent with the existing noise learning protocols that can characterize gate-dependent Pauli noise SPAM-robustly up to some degeneracy~\cite{erhard2019characterizing,hashim2020randomized,carignan2023error}. 
Our goal here is to find a lower bound for these scenarios.
More precisely, we hope to address the following question: What is the sample complexity to learn the identifiable parameters rather than the whole Pauli channel?

We will focus on a further simplified task: Given a partition of $n$-qubit Pauli operators into some disjoint sets, $\{S_i\}_i$. The task is to learn the geometric mean of the Pauli eigenvalues within each $S_i$ to additive precision $\varepsilon$ with high probability. We denote the maximal cardinality among all $S_i$ by $C$. 
As a motivating example, for the Pauli noise associated with a CNOT gate, it is shown~\cite{chen2023learnability} that the Pauli fidelity $\lambda_{XI},\lambda_{XX}$ cannot be identified individually, but the geometric mean $\sqrt{\lambda_{XI}\lambda_{XX}}$ can. 
All the other learnable parameters are also geometric mean of up to two Pauli fidelities (related to the fact that $\mr{CNOT}^2=\mathds1$). 
Note that, there can be more learnable parameters than those decided by a partition (\textit{e.g.}, $\sqrt{\lambda_{XI}\lambda_{XX}},\sqrt{\lambda_{YI}\lambda_{YX}},\sqrt{\lambda_{XI}\lambda_{YX}}$ are three independent learnable parameters), but this simplified task is sufficient to give a sample complexity lower bound.
We have the following result.
\begin{theorem}\label{th:coarse}
    Given a partition of the $n$-qubit Pauli group, $\{S_i\}_i$, with maximum cardinality $C\coleq\max_i|S_i|$, 
    define the geometrically-averaged Pauli fidelity $$\bar{\lambda}_{S_i}\coleq\mr{sgn}(\prod_{a\in S_i}\lambda_i)\prod_{a\in S_i}\left|\lambda_i\right|^{1/{|S_i|}}.$$    
    If there exists an entanglement-free scheme that, for any $n$-qubit Pauli channel $\Lambda$, outputs an estimator $\widehat\lambda_{S_i}$ such that $|\widehat\lambda_{S_i}-\bar\lambda_{S_i}|\le\varepsilon\le1/6C$ with probability at least $2/3$ for any $S_i$,
    after making $N$ rounds of measurements, then $N=\Omega({2^n}\varepsilon^{-2}{C^{-2}})$.
\end{theorem}
Many multi-qubit Clifford gates of practical interest have a polynomial (\textit{e.g.}, permutation gate)  or constant (\textit{e.g.}, parallel CNOTs) order. That is, applying the gate a polynomial/constant number of times yields identity. This means their learnable parameters are groups of at most polynomial/constant many Pauli eigenvalues.
Our result shows that there is still an exponential sample complexity lower bound for entanglement-free learning schemes in such cases. 

The techniques for proving Theorem~\ref{th:coarse} is very similar to those for Theorem~\ref{th:main}. Here, we just need to construct a different family of Pauli channels that can be distinguished by only looking at the geometrically averaged Pauli eigenvalues within each $S_i$. The other steps will carry over.
One may notice that Theorem~\ref{th:main} can be viewed as a corollary of Theorem~\ref{th:coarse}. We decide to present them separately for clarity. 
The proof is given in SM Sec.~\ref{app:degeneracy}.

\bigskip
\section{Discussion}
In this work, we introduce two classes of learning schemes to capture the notion of learning without entanglement. One is quantum resource-theoretic, using only entanglement non-generating operations between system and ancilla. The other is operational, describing quantum circuits assisted by mid-circuit measurement and classical feedforward. Both schemes are shown to be equivalent in terms of sample complexity.
We then prove a tight lower bound for Pauli channel learning within this model. Our results extend existing proof techniques~\cite{huang2022quantum,chen2022exponential} and improve upon the best-known lower bounds in the literature~\cite{chen2022quantum}. We also generalize our bounds for practical quantum noise characterization settings.

Our scenario differs from existing frameworks of learning with or without quantum memory, which we briefly review in below. For learning properties of quantum states, learning with quantum memory (or, quantum-enhanced learning) usually means one can perform collective measurement on multiple copies of the state, while learning without quantum memory (or, conventional learning) means only measurements on individual copies are allowed, though adaptivity is often allowed. Examples include \cite{bubeck2020entanglement, huang2021information, huang2022quantum,chen2022exponential,chen2022tight}. For learning properties of quantum channels, there have been multiple definitions of learning without memory: Refs.~\cite{huang2021information,aharonov2022quantum} requires the scheme to have no ancilla nor concatenation (i.e., sequentially applying the channel of interest); Refs.~\cite{huang2022quantum,chen2022exponential,caro2022learning} \textit{etc.} allow ancilla but not concatenation; In contrast, Ref.~\cite{chen2022quantum} studies schemes with concatenation but without ancilla. The scenario in the current paper is strictly more general than Ref.~\cite{chen2022quantum}, as we allow mid-circuit processing with quantum instruments instead of only quantum channels, and we justify our definition by connecting to the resource theory of entanglement.
It is interesting to explore other learning tasks that admit a separation in our definition.

The problem of Pauli channel learning has been studied with different figures of merit. 
For Pauli error rates, Ref.~\cite{flammia2020efficient} provides an ancilla-free protocol that learns the Pauli error rates to precision $\varepsilon$ in $l_2$-distance with $\widetilde O(2^n/\varepsilon^2)$ samples, which implies an upper bound of $\widetilde O(2^{3n}/\varepsilon^2)$ for learning in $l_1$-distance.
Ref. \cite{flammia2021pauli} shows that $\widetilde O(\log n/\varepsilon^2)$ samples is sufficient to learn the Pauli error rates to precision $\varepsilon$ in $l_\infty$-distance without using ancilla; In the case that ancilla is allowed, one can use Bell states and Bell measurements to directly sample from the Pauli error rates, which implies an $O(2^{2n}/\varepsilon^2)$ upper bound for learning in $l_1$-distance~\cite{flammia2021pauli,chen2022quantum};
For Pauli eigenvalues, Ref.~\cite{chen2022quantum} gives a family of $k$-qubit ancilla-assisted protocols using $O(n2^{n-k}/\varepsilon^2)$ samples for learning in $l_\infty$-distance, for $0\le k\le n$.
In terms of the lower bounds, Ref.~\cite{chen2022quantum} focuses on learning Pauli eigenvalues to constant error in $l_\infty$-distance, and in particular obtains a lower bound $\Omega(2^{n/3})$ for the number of measurements for any ancilla-free schemes with concatenation. %
The current work improves this lower bound to be tight (for more general schemes).
Another recent work studies learning Pauli error rate to error $\varepsilon$ in $l_1$-distance~\cite{fawzi2023lower}. For ancilla-free schemes with adaptivity, they obtains $\Omega(2^{2n}/\varepsilon^2)$ for general case and $\Omega(2^{2.5n}/\varepsilon^2)$ when $\varepsilon$ is exponentially small in $n$.
They allows concatenation with \emph{unital} processing channels and the lower bound holds for the number of measurements.
The results of the current work and Ref.~\cite{fawzi2023lower} do not imply each other~\footnote{We remark that, a lower bound of $\Omega(2^{3n}/\varepsilon^2)$ for learning Pauli error rates in $l_1$-distance would imply a bound of $\Omega(2^{n}/\varepsilon^2)$ for learning Pauli eigenvalues in $l_\infty$-distance, via the Parseval–Plancherel identity and relations of $l_p$-norm.}. Whether a tighter lower bound can be established with other figures of merit remains an open problem.

Our results have implications in practical quantum noise characterization tasks. On the one hand, it sets up an exponential barrier for any entanglement-free Pauli channel learning protocols~\cite{flammia2020efficient,erhard2019characterizing}, without additional assumptions on the Pauli noise model. The barrier persists even if one is only aimed at learning the SPAM-independently identifiable part of the noise channel. 
On the other hand, this motivates the development of an entanglement-assisted noise characterization protocol, which is pioneered in Ref.~\cite{chen2022quantum}. It is shown there that an entanglement-assisted protocol can learn the Pauli eigenvalues efficiently and SPAM-robustly, given access to good quantum memory. We believe the tight bounds obtained in the current work will strengthen the foundation for experimentally demonstrating the advantage of entanglement in this noise characterization task.

Finally, it is interesting to explore a deeper connection between the resource theory of entanglement~\cite{bennett1996concentrating,bennett1996mixed,vedral1997quantifying} and quantum learning theory. Specifically, since the tight bounds for Pauli channel learning with no entanglement and arbitrary entanglement have been settled, it is natural to ask what about learning with a certain amount of entanglement. There could be different ways of defining the entanglement cost of a learning scheme, one of which is to allow a bounded amount of ancillary qubits~\cite{chen2022quantum,chen2022exponential}. An upper bound of $\widetilde O(2^{n-k}/\varepsilon^2)$ for $k$-ancillary-qubit-assisted scheme is known and proved optimal for some restricted class of schemes~\cite{chen2022quantum}, but a general answer to this question is yet to be found.

\bigskip
\textit{Notes added.\textemdash}
During the completion of this manuscript, we are aware of an independent and contemporaneous work~\cite{chen2023futility} that obtains, among other results, a tight lower bound of $\Omega(2^n/\varepsilon^2)$ on the number of measurements for learning every Pauli eigenvalues to $\varepsilon$ precision with ancilla-free concatenating schemes.
Their proof is based on a different technique, which leads to different features in their results compared to ours:
On the one hand, the $\varepsilon$ in their bound can be any value within $(0,1]$, while ours is restricted to $\varepsilon\in(0, 1/6]$; On the other hand, our bound has a better constant factor and applies for the more general setting of classical-memory-assisted schemes.

\begin{acknowledgments}
    We thank Matthias Caro, Steve Flammia, John Preskill, Alireza Seif for helpful discussions. 
    We thank Sitan Chen, Weiyuan Gong for communicating with us about their independent and comtemporaneous results.
    {S.C., C.O., L.J. acknowledge support from the ARO(W911NF-23-1-0077), ARO MURI (W911NF-21-1-0325), AFOSR MURI (FA9550-19-1-0399, FA9550-21-1-0209), NSF (OMA-1936118, ERC-1941583, OMA-2137642), NTT Research, and the Packard Foundation (2020-71479). This material is based upon work supported by the U.S. Department of Energy, Office of Science, National Quantum Information Science Research Centers.
    S.C. thanks Caltech IQIM for hospitality, where part of this work is completed.}
    S.Z. acknowledges funding provided by the Institute for Quantum Information and Matter, an NSF Physics Frontiers Center (NSF Grant PHY-1733907) and Perimeter Institute for Theoretical Physics, a research institute supported in part by the Government of Canada through the Department of Innovation, Science and Economic Development Canada and by the Province of Ontario through the Ministry of Colleges and Universities.
    H.H. is supported by a Google PhD fellowship and a MediaTek Research Young Scholarship.
    H.H. acknowledges the visiting associate position at Massachusetts Institute of Technology.
\end{acknowledgments}

\bibliography{Pauli.bib}

\begin{thebibliography}{51}%
\makeatletter
\providecommand \@ifxundefined [1]{%
 \@ifx{#1\undefined}
}%
\providecommand \@ifnum [1]{%
 \ifnum #1\expandafter \@firstoftwo
 \else \expandafter \@secondoftwo
 \fi
}%
\providecommand \@ifx [1]{%
 \ifx #1\expandafter \@firstoftwo
 \else \expandafter \@secondoftwo
 \fi
}%
\providecommand \natexlab [1]{#1}%
\providecommand \enquote  [1]{``#1''}%
\providecommand \bibnamefont  [1]{#1}%
\providecommand \bibfnamefont [1]{#1}%
\providecommand \citenamefont [1]{#1}%
\providecommand \href@noop [0]{\@secondoftwo}%
\providecommand \href [0]{\begingroup \@sanitize@url \@href}%
\providecommand \@href[1]{\@@startlink{#1}\@@href}%
\providecommand \@@href[1]{\endgroup#1\@@endlink}%
\providecommand \@sanitize@url [0]{\catcode `\\12\catcode `\$12\catcode `\&12\catcode `\#12\catcode `\^12\catcode `\_12\catcode `\%12\relax}%
\providecommand \@@startlink[1]{}%
\providecommand \@@endlink[0]{}%
\providecommand \url  [0]{\begingroup\@sanitize@url \@url }%
\providecommand \@url [1]{\endgroup\@href {#1}{\urlprefix }}%
\providecommand \urlprefix  [0]{URL }%
\providecommand \Eprint [0]{\href }%
\providecommand \doibase [0]{https://doi.org/}%
\providecommand \selectlanguage [0]{\@gobble}%
\providecommand \bibinfo  [0]{\@secondoftwo}%
\providecommand \bibfield  [0]{\@secondoftwo}%
\providecommand \translation [1]{[#1]}%
\providecommand \BibitemOpen [0]{}%
\providecommand \bibitemStop [0]{}%
\providecommand \bibitemNoStop [0]{.\EOS\space}%
\providecommand \EOS [0]{\spacefactor3000\relax}%
\providecommand \BibitemShut  [1]{\csname bibitem#1\endcsname}%
\let\auto@bib@innerbib\@empty
%</preamble>
\bibitem [{\citenamefont {Nielsen}\ and\ \citenamefont {Chuang}(2011)}]{nielsen2002quantum}%
  \BibitemOpen
  \bibfield  {author} {\bibinfo {author} {\bibfnamefont {M.~A.}\ \bibnamefont {Nielsen}}\ and\ \bibinfo {author} {\bibfnamefont {I.~L.}\ \bibnamefont {Chuang}},\ }\href@noop {} {\emph {\bibinfo {title} {Quantum Computation and Quantum Information: 10th Anniversary Edition}}},\ \bibinfo {edition} {10th}\ ed.\ (\bibinfo  {publisher} {Cambridge University Press},\ \bibinfo {address} {USA},\ \bibinfo {year} {2011})\BibitemShut {NoStop}%
\bibitem [{\citenamefont {Gisin}\ and\ \citenamefont {Thew}(2007)}]{gisin2007quantum}%
  \BibitemOpen
  \bibfield  {author} {\bibinfo {author} {\bibfnamefont {N.}~\bibnamefont {Gisin}}\ and\ \bibinfo {author} {\bibfnamefont {R.}~\bibnamefont {Thew}},\ }\bibfield  {title} {\bibinfo {title} {Quantum communication},\ }\href@noop {} {\bibfield  {journal} {\bibinfo  {journal} {Nature photonics}\ }\textbf {\bibinfo {volume} {1}},\ \bibinfo {pages} {165} (\bibinfo {year} {2007})}\BibitemShut {NoStop}%
\bibitem [{\citenamefont {Kimble}(2008)}]{kimble2008quantum}%
  \BibitemOpen
  \bibfield  {author} {\bibinfo {author} {\bibfnamefont {H.~J.}\ \bibnamefont {Kimble}},\ }\bibfield  {title} {\bibinfo {title} {The quantum internet},\ }\href@noop {} {\bibfield  {journal} {\bibinfo  {journal} {Nature}\ }\textbf {\bibinfo {volume} {453}},\ \bibinfo {pages} {1023} (\bibinfo {year} {2008})}\BibitemShut {NoStop}%
\bibitem [{\citenamefont {Giovannetti}\ \emph {et~al.}(2006)\citenamefont {Giovannetti}, \citenamefont {Lloyd},\ and\ \citenamefont {Maccone}}]{giovannetti2006quantum}%
  \BibitemOpen
  \bibfield  {author} {\bibinfo {author} {\bibfnamefont {V.}~\bibnamefont {Giovannetti}}, \bibinfo {author} {\bibfnamefont {S.}~\bibnamefont {Lloyd}},\ and\ \bibinfo {author} {\bibfnamefont {L.}~\bibnamefont {Maccone}},\ }\bibfield  {title} {\bibinfo {title} {Quantum metrology},\ }\href@noop {} {\bibfield  {journal} {\bibinfo  {journal} {Physical review letters}\ }\textbf {\bibinfo {volume} {96}},\ \bibinfo {pages} {010401} (\bibinfo {year} {2006})}\BibitemShut {NoStop}%
\bibitem [{\citenamefont {Giovannetti}\ \emph {et~al.}(2011)\citenamefont {Giovannetti}, \citenamefont {Lloyd},\ and\ \citenamefont {Maccone}}]{giovannetti2011advances}%
  \BibitemOpen
  \bibfield  {author} {\bibinfo {author} {\bibfnamefont {V.}~\bibnamefont {Giovannetti}}, \bibinfo {author} {\bibfnamefont {S.}~\bibnamefont {Lloyd}},\ and\ \bibinfo {author} {\bibfnamefont {L.}~\bibnamefont {Maccone}},\ }\bibfield  {title} {\bibinfo {title} {Advances in quantum metrology},\ }\href@noop {} {\bibfield  {journal} {\bibinfo  {journal} {Nature photonics}\ }\textbf {\bibinfo {volume} {5}},\ \bibinfo {pages} {222} (\bibinfo {year} {2011})}\BibitemShut {NoStop}%
\bibitem [{\citenamefont {Polino}\ \emph {et~al.}(2020)\citenamefont {Polino}, \citenamefont {Valeri}, \citenamefont {Spagnolo},\ and\ \citenamefont {Sciarrino}}]{polino2020photonic}%
  \BibitemOpen
  \bibfield  {author} {\bibinfo {author} {\bibfnamefont {E.}~\bibnamefont {Polino}}, \bibinfo {author} {\bibfnamefont {M.}~\bibnamefont {Valeri}}, \bibinfo {author} {\bibfnamefont {N.}~\bibnamefont {Spagnolo}},\ and\ \bibinfo {author} {\bibfnamefont {F.}~\bibnamefont {Sciarrino}},\ }\bibfield  {title} {\bibinfo {title} {Photonic quantum metrology},\ }\href@noop {} {\bibfield  {journal} {\bibinfo  {journal} {AVS Quantum Science}\ }\textbf {\bibinfo {volume} {2}},\ \bibinfo {pages} {024703} (\bibinfo {year} {2020})}\BibitemShut {NoStop}%
\bibitem [{\citenamefont {Huang}\ \emph {et~al.}(2022)\citenamefont {Huang}, \citenamefont {Broughton}, \citenamefont {Cotler}, \citenamefont {Chen}, \citenamefont {Li}, \citenamefont {Mohseni}, \citenamefont {Neven}, \citenamefont {Babbush}, \citenamefont {Kueng}, \citenamefont {Preskill} \emph {et~al.}}]{huang2022quantum}%
  \BibitemOpen
  \bibfield  {author} {\bibinfo {author} {\bibfnamefont {H.-Y.}\ \bibnamefont {Huang}}, \bibinfo {author} {\bibfnamefont {M.}~\bibnamefont {Broughton}}, \bibinfo {author} {\bibfnamefont {J.}~\bibnamefont {Cotler}}, \bibinfo {author} {\bibfnamefont {S.}~\bibnamefont {Chen}}, \bibinfo {author} {\bibfnamefont {J.}~\bibnamefont {Li}}, \bibinfo {author} {\bibfnamefont {M.}~\bibnamefont {Mohseni}}, \bibinfo {author} {\bibfnamefont {H.}~\bibnamefont {Neven}}, \bibinfo {author} {\bibfnamefont {R.}~\bibnamefont {Babbush}}, \bibinfo {author} {\bibfnamefont {R.}~\bibnamefont {Kueng}}, \bibinfo {author} {\bibfnamefont {J.}~\bibnamefont {Preskill}}, \emph {et~al.},\ }\bibfield  {title} {\bibinfo {title} {Quantum advantage in learning from experiments},\ }\href {https://doi.org/10.1126/science.abn7293} {\bibfield  {journal} {\bibinfo  {journal} {Science}\ }\textbf {\bibinfo {volume} {376}},\ \bibinfo {pages} {1182} (\bibinfo {year} {2022})}\BibitemShut {NoStop}%
\bibitem [{\citenamefont {Chen}\ \emph {et~al.}(2022{\natexlab{a}})\citenamefont {Chen}, \citenamefont {Cotler}, \citenamefont {Huang},\ and\ \citenamefont {Li}}]{chen2022exponential}%
  \BibitemOpen
  \bibfield  {author} {\bibinfo {author} {\bibfnamefont {S.}~\bibnamefont {Chen}}, \bibinfo {author} {\bibfnamefont {J.}~\bibnamefont {Cotler}}, \bibinfo {author} {\bibfnamefont {H.-Y.}\ \bibnamefont {Huang}},\ and\ \bibinfo {author} {\bibfnamefont {J.}~\bibnamefont {Li}},\ }\bibfield  {title} {\bibinfo {title} {Exponential separations between learning with and without quantum memory},\ }in\ \href {https://doi.org/10.1109/FOCS52979.2021.00063} {\emph {\bibinfo {booktitle} {2021 IEEE 62nd Annual Symposium on Foundations of Computer Science (FOCS)}}}\ (\bibinfo {organization} {IEEE},\ \bibinfo {year} {2022})\ pp.\ \bibinfo {pages} {574--585}\BibitemShut {NoStop}%
\bibitem [{\citenamefont {Bubeck}\ \emph {et~al.}(2020)\citenamefont {Bubeck}, \citenamefont {Chen},\ and\ \citenamefont {Li}}]{bubeck2020entanglement}%
  \BibitemOpen
  \bibfield  {author} {\bibinfo {author} {\bibfnamefont {S.}~\bibnamefont {Bubeck}}, \bibinfo {author} {\bibfnamefont {S.}~\bibnamefont {Chen}},\ and\ \bibinfo {author} {\bibfnamefont {J.}~\bibnamefont {Li}},\ }\bibfield  {title} {\bibinfo {title} {Entanglement is necessary for optimal quantum property testing},\ }in\ \href {https://doi.org/10.1109/FOCS46700.2020.00070} {\emph {\bibinfo {booktitle} {2020 IEEE 61st Annual Symposium on Foundations of Computer Science (FOCS)}}}\ (\bibinfo {organization} {IEEE},\ \bibinfo {year} {2020})\ pp.\ \bibinfo {pages} {692--703}\BibitemShut {NoStop}%
\bibitem [{\citenamefont {Huang}\ \emph {et~al.}(2021)\citenamefont {Huang}, \citenamefont {Kueng},\ and\ \citenamefont {Preskill}}]{huang2021information}%
  \BibitemOpen
  \bibfield  {author} {\bibinfo {author} {\bibfnamefont {H.-Y.}\ \bibnamefont {Huang}}, \bibinfo {author} {\bibfnamefont {R.}~\bibnamefont {Kueng}},\ and\ \bibinfo {author} {\bibfnamefont {J.}~\bibnamefont {Preskill}},\ }\bibfield  {title} {\bibinfo {title} {Information-theoretic bounds on quantum advantage in machine learning},\ }\href {https://doi.org/10.1103/PhysRevLett.126.190505} {\bibfield  {journal} {\bibinfo  {journal} {Phys. Rev. Lett.}\ }\textbf {\bibinfo {volume} {126}},\ \bibinfo {pages} {190505} (\bibinfo {year} {2021})}\BibitemShut {NoStop}%
\bibitem [{\citenamefont {Chen}\ \emph {et~al.}(2022{\natexlab{b}})\citenamefont {Chen}, \citenamefont {Huang}, \citenamefont {Li}, \citenamefont {Liu},\ and\ \citenamefont {Sellke}}]{chen2022tight}%
  \BibitemOpen
  \bibfield  {author} {\bibinfo {author} {\bibfnamefont {S.}~\bibnamefont {Chen}}, \bibinfo {author} {\bibfnamefont {B.}~\bibnamefont {Huang}}, \bibinfo {author} {\bibfnamefont {J.}~\bibnamefont {Li}}, \bibinfo {author} {\bibfnamefont {A.}~\bibnamefont {Liu}},\ and\ \bibinfo {author} {\bibfnamefont {M.}~\bibnamefont {Sellke}},\ }\bibfield  {title} {\bibinfo {title} {Tight bounds for state tomography with incoherent measurements},\ }\bibfield  {journal} {\bibinfo  {journal} {arXiv preprint arXiv:2206.05265}\ }\href {https://doi.org/10.48550/arXiv.2206.05265} {10.48550/arXiv.2206.05265} (\bibinfo {year} {2022}{\natexlab{b}})\BibitemShut {NoStop}%
\bibitem [{\citenamefont {Aharonov}\ \emph {et~al.}(2022)\citenamefont {Aharonov}, \citenamefont {Cotler},\ and\ \citenamefont {Qi}}]{aharonov2022quantum}%
  \BibitemOpen
  \bibfield  {author} {\bibinfo {author} {\bibfnamefont {D.}~\bibnamefont {Aharonov}}, \bibinfo {author} {\bibfnamefont {J.}~\bibnamefont {Cotler}},\ and\ \bibinfo {author} {\bibfnamefont {X.-L.}\ \bibnamefont {Qi}},\ }\bibfield  {title} {\bibinfo {title} {Quantum algorithmic measurement},\ }\href {https://doi.org/10.1038/s41467-021-27922-0} {\bibfield  {journal} {\bibinfo  {journal} {Nature communications}\ }\textbf {\bibinfo {volume} {13}},\ \bibinfo {pages} {887} (\bibinfo {year} {2022})}\BibitemShut {NoStop}%
\bibitem [{\citenamefont {Caro}(2022)}]{caro2022learning}%
  \BibitemOpen
  \bibfield  {author} {\bibinfo {author} {\bibfnamefont {M.~C.}\ \bibnamefont {Caro}},\ }\bibfield  {title} {\bibinfo {title} {Learning quantum processes and hamiltonians via the pauli transfer matrix},\ }\bibfield  {journal} {\bibinfo  {journal} {arXiv preprint arXiv:2212.04471}\ }\href {https://doi.org/10.48550/arXiv.2212.04471} {10.48550/arXiv.2212.04471} (\bibinfo {year} {2022})\BibitemShut {NoStop}%
\bibitem [{\citenamefont {Chen}\ \emph {et~al.}(2022{\natexlab{c}})\citenamefont {Chen}, \citenamefont {Zhou}, \citenamefont {Seif},\ and\ \citenamefont {Jiang}}]{chen2022quantum}%
  \BibitemOpen
  \bibfield  {author} {\bibinfo {author} {\bibfnamefont {S.}~\bibnamefont {Chen}}, \bibinfo {author} {\bibfnamefont {S.}~\bibnamefont {Zhou}}, \bibinfo {author} {\bibfnamefont {A.}~\bibnamefont {Seif}},\ and\ \bibinfo {author} {\bibfnamefont {L.}~\bibnamefont {Jiang}},\ }\bibfield  {title} {\bibinfo {title} {Quantum advantages for {P}auli channel estimation},\ }\href {https://doi.org/10.1103/PhysRevA.105.032435} {\bibfield  {journal} {\bibinfo  {journal} {Physical Review A}\ }\textbf {\bibinfo {volume} {105}},\ \bibinfo {pages} {032435} (\bibinfo {year} {2022}{\natexlab{c}})}\BibitemShut {NoStop}%
\bibitem [{\citenamefont {Chen}\ \emph {et~al.}(2023{\natexlab{a}})\citenamefont {Chen}, \citenamefont {Cotler}, \citenamefont {Huang},\ and\ \citenamefont {Li}}]{chen2023complexity}%
  \BibitemOpen
  \bibfield  {author} {\bibinfo {author} {\bibfnamefont {S.}~\bibnamefont {Chen}}, \bibinfo {author} {\bibfnamefont {J.}~\bibnamefont {Cotler}}, \bibinfo {author} {\bibfnamefont {H.-Y.}\ \bibnamefont {Huang}},\ and\ \bibinfo {author} {\bibfnamefont {J.}~\bibnamefont {Li}},\ }\bibfield  {title} {\bibinfo {title} {The complexity of nisq},\ }\href@noop {} {\bibfield  {journal} {\bibinfo  {journal} {Nature Communications}\ }\textbf {\bibinfo {volume} {14}},\ \bibinfo {pages} {6001} (\bibinfo {year} {2023}{\natexlab{a}})}\BibitemShut {NoStop}%
\bibitem [{\citenamefont {Chitambar}\ and\ \citenamefont {Gour}(2019)}]{chitambar2019quantum}%
  \BibitemOpen
  \bibfield  {author} {\bibinfo {author} {\bibfnamefont {E.}~\bibnamefont {Chitambar}}\ and\ \bibinfo {author} {\bibfnamefont {G.}~\bibnamefont {Gour}},\ }\bibfield  {title} {\bibinfo {title} {Quantum resource theories},\ }\href {https://doi.org/10.1103/RevModPhys.91.025001} {\bibfield  {journal} {\bibinfo  {journal} {Reviews of Modern Physics}\ }\textbf {\bibinfo {volume} {91}},\ \bibinfo {pages} {025001} (\bibinfo {year} {2019})}\BibitemShut {NoStop}%
\bibitem [{\citenamefont {Shor}(1996)}]{shor1996fault}%
  \BibitemOpen
  \bibfield  {author} {\bibinfo {author} {\bibfnamefont {P.~W.}\ \bibnamefont {Shor}},\ }\bibfield  {title} {\bibinfo {title} {Fault-tolerant quantum computation},\ }in\ \href {https://doi.org/10.1109/SFCS.1996.548464} {\emph {\bibinfo {booktitle} {Proceedings of 37th Conference on Foundations of Computer Science}}}\ (\bibinfo {organization} {IEEE},\ \bibinfo {year} {1996})\ pp.\ \bibinfo {pages} {56--65}\BibitemShut {NoStop}%
\bibitem [{\citenamefont {Gottesman}(1998)}]{gottesman1998theory}%
  \BibitemOpen
  \bibfield  {author} {\bibinfo {author} {\bibfnamefont {D.}~\bibnamefont {Gottesman}},\ }\bibfield  {title} {\bibinfo {title} {Theory of fault-tolerant quantum computation},\ }\href {https://doi.org/10.1103/PhysRevA.57.127} {\bibfield  {journal} {\bibinfo  {journal} {Physical Review A}\ }\textbf {\bibinfo {volume} {57}},\ \bibinfo {pages} {127} (\bibinfo {year} {1998})}\BibitemShut {NoStop}%
\bibitem [{\citenamefont {Iqbal}\ \emph {et~al.}(2023)\citenamefont {Iqbal}, \citenamefont {Tantivasadakarn}, \citenamefont {Gatterman}, \citenamefont {Gerber}, \citenamefont {Gilmore}, \citenamefont {Gresh}, \citenamefont {Hankin}, \citenamefont {Hewitt}, \citenamefont {Horst}, \citenamefont {Matheny} \emph {et~al.}}]{iqbal2023topological}%
  \BibitemOpen
  \bibfield  {author} {\bibinfo {author} {\bibfnamefont {M.}~\bibnamefont {Iqbal}}, \bibinfo {author} {\bibfnamefont {N.}~\bibnamefont {Tantivasadakarn}}, \bibinfo {author} {\bibfnamefont {T.~M.}\ \bibnamefont {Gatterman}}, \bibinfo {author} {\bibfnamefont {J.~A.}\ \bibnamefont {Gerber}}, \bibinfo {author} {\bibfnamefont {K.}~\bibnamefont {Gilmore}}, \bibinfo {author} {\bibfnamefont {D.}~\bibnamefont {Gresh}}, \bibinfo {author} {\bibfnamefont {A.}~\bibnamefont {Hankin}}, \bibinfo {author} {\bibfnamefont {N.}~\bibnamefont {Hewitt}}, \bibinfo {author} {\bibfnamefont {C.~V.}\ \bibnamefont {Horst}}, \bibinfo {author} {\bibfnamefont {M.}~\bibnamefont {Matheny}}, \emph {et~al.},\ }\bibfield  {title} {\bibinfo {title} {Topological order from measurements and feed-forward on a trapped ion quantum computer},\ }\bibfield  {journal} {\bibinfo  {journal} {arXiv preprint arXiv:2302.01917}\ }\href {https://doi.org/10.48550/arXiv.2302.01917} {10.48550/arXiv.2302.01917} (\bibinfo {year} {2023})\BibitemShut {NoStop}%
\bibitem [{\citenamefont {Singh}\ \emph {et~al.}(2023)\citenamefont {Singh}, \citenamefont {Bradley}, \citenamefont {Anand}, \citenamefont {Ramesh}, \citenamefont {White},\ and\ \citenamefont {Bernien}}]{singh2023mid}%
  \BibitemOpen
  \bibfield  {author} {\bibinfo {author} {\bibfnamefont {K.}~\bibnamefont {Singh}}, \bibinfo {author} {\bibfnamefont {C.}~\bibnamefont {Bradley}}, \bibinfo {author} {\bibfnamefont {S.}~\bibnamefont {Anand}}, \bibinfo {author} {\bibfnamefont {V.}~\bibnamefont {Ramesh}}, \bibinfo {author} {\bibfnamefont {R.}~\bibnamefont {White}},\ and\ \bibinfo {author} {\bibfnamefont {H.}~\bibnamefont {Bernien}},\ }\bibfield  {title} {\bibinfo {title} {Mid-circuit correction of correlated phase errors using an array of spectator qubits},\ }\href {https://doi.org/10.1126/science.ade5337} {\bibfield  {journal} {\bibinfo  {journal} {Science}\ }\textbf {\bibinfo {volume} {380}},\ \bibinfo {pages} {1265} (\bibinfo {year} {2023})}\BibitemShut {NoStop}%
\bibitem [{\citenamefont {Erhard}\ \emph {et~al.}(2019)\citenamefont {Erhard}, \citenamefont {Wallman}, \citenamefont {Postler}, \citenamefont {Meth}, \citenamefont {Stricker}, \citenamefont {Martinez}, \citenamefont {Schindler}, \citenamefont {Monz}, \citenamefont {Emerson},\ and\ \citenamefont {Blatt}}]{erhard2019characterizing}%
  \BibitemOpen
  \bibfield  {author} {\bibinfo {author} {\bibfnamefont {A.}~\bibnamefont {Erhard}}, \bibinfo {author} {\bibfnamefont {J.~J.}\ \bibnamefont {Wallman}}, \bibinfo {author} {\bibfnamefont {L.}~\bibnamefont {Postler}}, \bibinfo {author} {\bibfnamefont {M.}~\bibnamefont {Meth}}, \bibinfo {author} {\bibfnamefont {R.}~\bibnamefont {Stricker}}, \bibinfo {author} {\bibfnamefont {E.~A.}\ \bibnamefont {Martinez}}, \bibinfo {author} {\bibfnamefont {P.}~\bibnamefont {Schindler}}, \bibinfo {author} {\bibfnamefont {T.}~\bibnamefont {Monz}}, \bibinfo {author} {\bibfnamefont {J.}~\bibnamefont {Emerson}},\ and\ \bibinfo {author} {\bibfnamefont {R.}~\bibnamefont {Blatt}},\ }\bibfield  {title} {\bibinfo {title} {Characterizing large-scale quantum computers via cycle benchmarking},\ }\href {https://doi.org/10.1038/s41467-019-13068-7} {\bibfield  {journal} {\bibinfo  {journal} {Nature Communications}\ }\textbf {\bibinfo {volume} {10}},\ \bibinfo {pages} {5347} (\bibinfo {year} {2019})}\BibitemShut {NoStop}%
\bibitem [{\citenamefont {Carignan-Dugas}\ \emph {et~al.}(2023)\citenamefont {Carignan-Dugas}, \citenamefont {Dahlen}, \citenamefont {Hincks}, \citenamefont {Ospadov}, \citenamefont {Beale}, \citenamefont {Ferracin}, \citenamefont {Skanes-Norman}, \citenamefont {Emerson},\ and\ \citenamefont {Wallman}}]{carignan2023error}%
  \BibitemOpen
  \bibfield  {author} {\bibinfo {author} {\bibfnamefont {A.}~\bibnamefont {Carignan-Dugas}}, \bibinfo {author} {\bibfnamefont {D.}~\bibnamefont {Dahlen}}, \bibinfo {author} {\bibfnamefont {I.}~\bibnamefont {Hincks}}, \bibinfo {author} {\bibfnamefont {E.}~\bibnamefont {Ospadov}}, \bibinfo {author} {\bibfnamefont {S.~J.}\ \bibnamefont {Beale}}, \bibinfo {author} {\bibfnamefont {S.}~\bibnamefont {Ferracin}}, \bibinfo {author} {\bibfnamefont {J.}~\bibnamefont {Skanes-Norman}}, \bibinfo {author} {\bibfnamefont {J.}~\bibnamefont {Emerson}},\ and\ \bibinfo {author} {\bibfnamefont {J.~J.}\ \bibnamefont {Wallman}},\ }\bibfield  {title} {\bibinfo {title} {The error reconstruction and compiled calibration of quantum computing cycles},\ }\bibfield  {journal} {\bibinfo  {journal} {arXiv preprint arXiv:2303.17714}\ }\href {https://doi.org/10.48550/arXiv.2303.17714} {10.48550/arXiv.2303.17714} (\bibinfo {year} {2023})\BibitemShut {NoStop}%
\bibitem [{\citenamefont {Harper}\ \emph {et~al.}(2020)\citenamefont {Harper}, \citenamefont {Flammia},\ and\ \citenamefont {Wallman}}]{harper2020efficient}%
  \BibitemOpen
  \bibfield  {author} {\bibinfo {author} {\bibfnamefont {R.}~\bibnamefont {Harper}}, \bibinfo {author} {\bibfnamefont {S.~T.}\ \bibnamefont {Flammia}},\ and\ \bibinfo {author} {\bibfnamefont {J.~J.}\ \bibnamefont {Wallman}},\ }\bibfield  {title} {\bibinfo {title} {Efficient learning of quantum noise},\ }\href {https://doi.org/10.1038/s41567-020-0992-8} {\bibfield  {journal} {\bibinfo  {journal} {Nature Physics}\ }\textbf {\bibinfo {volume} {16}},\ \bibinfo {pages} {1184} (\bibinfo {year} {2020})}\BibitemShut {NoStop}%
\bibitem [{\citenamefont {Van Den~Berg}\ \emph {et~al.}(2023)\citenamefont {Van Den~Berg}, \citenamefont {Minev}, \citenamefont {Kandala},\ and\ \citenamefont {Temme}}]{van2023probabilistic}%
  \BibitemOpen
  \bibfield  {author} {\bibinfo {author} {\bibfnamefont {E.}~\bibnamefont {Van Den~Berg}}, \bibinfo {author} {\bibfnamefont {Z.~K.}\ \bibnamefont {Minev}}, \bibinfo {author} {\bibfnamefont {A.}~\bibnamefont {Kandala}},\ and\ \bibinfo {author} {\bibfnamefont {K.}~\bibnamefont {Temme}},\ }\bibfield  {title} {\bibinfo {title} {Probabilistic error cancellation with sparse pauli--lindblad models on noisy quantum processors},\ }\bibfield  {journal} {\bibinfo  {journal} {Nature Physics}\ }\href {https://doi.org/10.1038/s41567-023-02042-2} {10.1038/s41567-023-02042-2} (\bibinfo {year} {2023})\BibitemShut {NoStop}%
\bibitem [{\citenamefont {{Ferracin}}\ \emph {et~al.}(2022)\citenamefont {{Ferracin}}, \citenamefont {{Hashim}}, \citenamefont {{Ville}}, \citenamefont {{Naik}}, \citenamefont {{Carignan-Dugas}}, \citenamefont {{Qassim}}, \citenamefont {{Morvan}}, \citenamefont {{Santiago}}, \citenamefont {{Siddiqi}},\ and\ \citenamefont {{Wallman}}}]{Ferracin2022Efficiently}%
  \BibitemOpen
  \bibfield  {author} {\bibinfo {author} {\bibfnamefont {S.}~\bibnamefont {{Ferracin}}}, \bibinfo {author} {\bibfnamefont {A.}~\bibnamefont {{Hashim}}}, \bibinfo {author} {\bibfnamefont {J.-L.}\ \bibnamefont {{Ville}}}, \bibinfo {author} {\bibfnamefont {R.}~\bibnamefont {{Naik}}}, \bibinfo {author} {\bibfnamefont {A.}~\bibnamefont {{Carignan-Dugas}}}, \bibinfo {author} {\bibfnamefont {H.}~\bibnamefont {{Qassim}}}, \bibinfo {author} {\bibfnamefont {A.}~\bibnamefont {{Morvan}}}, \bibinfo {author} {\bibfnamefont {D.~I.}\ \bibnamefont {{Santiago}}}, \bibinfo {author} {\bibfnamefont {I.}~\bibnamefont {{Siddiqi}}},\ and\ \bibinfo {author} {\bibfnamefont {J.~J.}\ \bibnamefont {{Wallman}}},\ }\href@noop {} {\bibinfo {title} {{Efficiently improving the performance of noisy quantum computers}}} (\bibinfo {year} {2022}),\ \Eprint {https://arxiv.org/abs/2201.10672} {arXiv:2201.10672 [quant-ph]} \BibitemShut {NoStop}%
\bibitem [{\citenamefont {Kim}\ \emph {et~al.}(2023)\citenamefont {Kim}, \citenamefont {Eddins}, \citenamefont {Anand}, \citenamefont {Wei}, \citenamefont {Van Den~Berg}, \citenamefont {Rosenblatt}, \citenamefont {Nayfeh}, \citenamefont {Wu}, \citenamefont {Zaletel}, \citenamefont {Temme} \emph {et~al.}}]{kim2023evidence}%
  \BibitemOpen
  \bibfield  {author} {\bibinfo {author} {\bibfnamefont {Y.}~\bibnamefont {Kim}}, \bibinfo {author} {\bibfnamefont {A.}~\bibnamefont {Eddins}}, \bibinfo {author} {\bibfnamefont {S.}~\bibnamefont {Anand}}, \bibinfo {author} {\bibfnamefont {K.~X.}\ \bibnamefont {Wei}}, \bibinfo {author} {\bibfnamefont {E.}~\bibnamefont {Van Den~Berg}}, \bibinfo {author} {\bibfnamefont {S.}~\bibnamefont {Rosenblatt}}, \bibinfo {author} {\bibfnamefont {H.}~\bibnamefont {Nayfeh}}, \bibinfo {author} {\bibfnamefont {Y.}~\bibnamefont {Wu}}, \bibinfo {author} {\bibfnamefont {M.}~\bibnamefont {Zaletel}}, \bibinfo {author} {\bibfnamefont {K.}~\bibnamefont {Temme}}, \emph {et~al.},\ }\bibfield  {title} {\bibinfo {title} {Evidence for the utility of quantum computing before fault tolerance},\ }\href {https://doi.org/10.1038/s41586-023-06096-3} {\bibfield  {journal} {\bibinfo  {journal} {Nature}\ }\textbf {\bibinfo {volume} {618}},\ \bibinfo {pages} {500} (\bibinfo {year} {2023})}\BibitemShut {NoStop}%
\bibitem [{\citenamefont {Tuckett}\ \emph {et~al.}(2018)\citenamefont {Tuckett}, \citenamefont {Bartlett},\ and\ \citenamefont {Flammia}}]{tuckett2018ultrahigh}%
  \BibitemOpen
  \bibfield  {author} {\bibinfo {author} {\bibfnamefont {D.~K.}\ \bibnamefont {Tuckett}}, \bibinfo {author} {\bibfnamefont {S.~D.}\ \bibnamefont {Bartlett}},\ and\ \bibinfo {author} {\bibfnamefont {S.~T.}\ \bibnamefont {Flammia}},\ }\bibfield  {title} {\bibinfo {title} {Ultrahigh error threshold for surface codes with biased noise},\ }\href {https://doi.org/10.1103/PhysRevLett.120.050505} {\bibfield  {journal} {\bibinfo  {journal} {Physical review letters}\ }\textbf {\bibinfo {volume} {120}},\ \bibinfo {pages} {050505} (\bibinfo {year} {2018})}\BibitemShut {NoStop}%
\bibitem [{\citenamefont {Wallman}\ and\ \citenamefont {Emerson}(2016)}]{wallman2016noise}%
  \BibitemOpen
  \bibfield  {author} {\bibinfo {author} {\bibfnamefont {J.~J.}\ \bibnamefont {Wallman}}\ and\ \bibinfo {author} {\bibfnamefont {J.}~\bibnamefont {Emerson}},\ }\bibfield  {title} {\bibinfo {title} {Noise tailoring for scalable quantum computation via randomized compiling},\ }\href {https://doi.org/10.1103/PhysRevA.94.052325} {\bibfield  {journal} {\bibinfo  {journal} {Physical Review A}\ }\textbf {\bibinfo {volume} {94}},\ \bibinfo {pages} {052325} (\bibinfo {year} {2016})}\BibitemShut {NoStop}%
\bibitem [{\citenamefont {Hashim}\ \emph {et~al.}(2021)\citenamefont {Hashim}, \citenamefont {Naik}, \citenamefont {Morvan}, \citenamefont {Ville}, \citenamefont {Mitchell}, \citenamefont {Kreikebaum}, \citenamefont {Davis}, \citenamefont {Smith}, \citenamefont {Iancu}, \citenamefont {O'Brien}, \citenamefont {Hincks}, \citenamefont {Wallman}, \citenamefont {Emerson},\ and\ \citenamefont {Siddiqi}}]{hashim2020randomized}%
  \BibitemOpen
  \bibfield  {author} {\bibinfo {author} {\bibfnamefont {A.}~\bibnamefont {Hashim}}, \bibinfo {author} {\bibfnamefont {R.~K.}\ \bibnamefont {Naik}}, \bibinfo {author} {\bibfnamefont {A.}~\bibnamefont {Morvan}}, \bibinfo {author} {\bibfnamefont {J.-L.}\ \bibnamefont {Ville}}, \bibinfo {author} {\bibfnamefont {B.}~\bibnamefont {Mitchell}}, \bibinfo {author} {\bibfnamefont {J.~M.}\ \bibnamefont {Kreikebaum}}, \bibinfo {author} {\bibfnamefont {M.}~\bibnamefont {Davis}}, \bibinfo {author} {\bibfnamefont {E.}~\bibnamefont {Smith}}, \bibinfo {author} {\bibfnamefont {C.}~\bibnamefont {Iancu}}, \bibinfo {author} {\bibfnamefont {K.~P.}\ \bibnamefont {O'Brien}}, \bibinfo {author} {\bibfnamefont {I.}~\bibnamefont {Hincks}}, \bibinfo {author} {\bibfnamefont {J.~J.}\ \bibnamefont {Wallman}}, \bibinfo {author} {\bibfnamefont {J.}~\bibnamefont {Emerson}},\ and\ \bibinfo {author} {\bibfnamefont {I.}~\bibnamefont {Siddiqi}},\ }\bibfield  {title} {\bibinfo {title} {Randomized compiling for scalable quantum computing on a noisy superconducting quantum processor},\ }\href {https://doi.org/10.1103/PhysRevX.11.041039} {\bibfield  {journal} {\bibinfo  {journal} {Phys. Rev. X}\ }\textbf {\bibinfo {volume} {11}},\ \bibinfo {pages} {041039} (\bibinfo {year} {2021})}\BibitemShut {NoStop}%
\bibitem [{\citenamefont {Flammia}\ and\ \citenamefont {Wallman}(2020)}]{flammia2020efficient}%
  \BibitemOpen
  \bibfield  {author} {\bibinfo {author} {\bibfnamefont {S.~T.}\ \bibnamefont {Flammia}}\ and\ \bibinfo {author} {\bibfnamefont {J.~J.}\ \bibnamefont {Wallman}},\ }\bibfield  {title} {\bibinfo {title} {Efficient estimation of {Pauli} channels},\ }\bibfield  {journal} {\bibinfo  {journal} {ACM Transactions on Quantum Computing}\ }\textbf {\bibinfo {volume} {1}},\ \href {https://doi.org/10.1145/3408039} {10.1145/3408039} (\bibinfo {year} {2020})\BibitemShut {NoStop}%
\bibitem [{\citenamefont {Flammia}\ and\ \citenamefont {O'Donnell}(2021)}]{flammia2021pauli}%
  \BibitemOpen
  \bibfield  {author} {\bibinfo {author} {\bibfnamefont {S.~T.}\ \bibnamefont {Flammia}}\ and\ \bibinfo {author} {\bibfnamefont {R.}~\bibnamefont {O'Donnell}},\ }\bibfield  {title} {\bibinfo {title} {Pauli error estimation via population recovery},\ }\href {https://doi.org/10.22331/q-2021-09-23-549} {\bibfield  {journal} {\bibinfo  {journal} {Quantum}\ }\textbf {\bibinfo {volume} {5}},\ \bibinfo {pages} {549} (\bibinfo {year} {2021})}\BibitemShut {NoStop}%
\bibitem [{\citenamefont {Fawzi}\ \emph {et~al.}(2023)\citenamefont {Fawzi}, \citenamefont {Oufkir},\ and\ \citenamefont {Fran{\c{c}}a}}]{fawzi2023lower}%
  \BibitemOpen
  \bibfield  {author} {\bibinfo {author} {\bibfnamefont {O.}~\bibnamefont {Fawzi}}, \bibinfo {author} {\bibfnamefont {A.}~\bibnamefont {Oufkir}},\ and\ \bibinfo {author} {\bibfnamefont {D.~S.}\ \bibnamefont {Fran{\c{c}}a}},\ }\bibfield  {title} {\bibinfo {title} {Lower bounds on learning pauli channels},\ }\bibfield  {journal} {\bibinfo  {journal} {arXiv preprint arXiv:2301.09192}\ }\href {https://doi.org/10.48550/arXiv.2301.09192} {10.48550/arXiv.2301.09192} (\bibinfo {year} {2023})\BibitemShut {NoStop}%
\bibitem [{\citenamefont {Bennett}\ \emph {et~al.}(1996{\natexlab{a}})\citenamefont {Bennett}, \citenamefont {Bernstein}, \citenamefont {Popescu},\ and\ \citenamefont {Schumacher}}]{bennett1996concentrating}%
  \BibitemOpen
  \bibfield  {author} {\bibinfo {author} {\bibfnamefont {C.~H.}\ \bibnamefont {Bennett}}, \bibinfo {author} {\bibfnamefont {H.~J.}\ \bibnamefont {Bernstein}}, \bibinfo {author} {\bibfnamefont {S.}~\bibnamefont {Popescu}},\ and\ \bibinfo {author} {\bibfnamefont {B.}~\bibnamefont {Schumacher}},\ }\bibfield  {title} {\bibinfo {title} {Concentrating partial entanglement by local operations},\ }\href {https://doi.org/10.1103/PhysRevA.53.2046} {\bibfield  {journal} {\bibinfo  {journal} {Physical Review A}\ }\textbf {\bibinfo {volume} {53}},\ \bibinfo {pages} {2046} (\bibinfo {year} {1996}{\natexlab{a}})}\BibitemShut {NoStop}%
\bibitem [{\citenamefont {Bennett}\ \emph {et~al.}(1996{\natexlab{b}})\citenamefont {Bennett}, \citenamefont {DiVincenzo}, \citenamefont {Smolin},\ and\ \citenamefont {Wootters}}]{bennett1996mixed}%
  \BibitemOpen
  \bibfield  {author} {\bibinfo {author} {\bibfnamefont {C.~H.}\ \bibnamefont {Bennett}}, \bibinfo {author} {\bibfnamefont {D.~P.}\ \bibnamefont {DiVincenzo}}, \bibinfo {author} {\bibfnamefont {J.~A.}\ \bibnamefont {Smolin}},\ and\ \bibinfo {author} {\bibfnamefont {W.~K.}\ \bibnamefont {Wootters}},\ }\bibfield  {title} {\bibinfo {title} {Mixed-state entanglement and quantum error correction},\ }\href {https://doi.org/10.1103/PhysRevA.54.3824} {\bibfield  {journal} {\bibinfo  {journal} {Physical Review A}\ }\textbf {\bibinfo {volume} {54}},\ \bibinfo {pages} {3824} (\bibinfo {year} {1996}{\natexlab{b}})}\BibitemShut {NoStop}%
\bibitem [{\citenamefont {Vedral}\ \emph {et~al.}(1997)\citenamefont {Vedral}, \citenamefont {Plenio}, \citenamefont {Rippin},\ and\ \citenamefont {Knight}}]{vedral1997quantifying}%
  \BibitemOpen
  \bibfield  {author} {\bibinfo {author} {\bibfnamefont {V.}~\bibnamefont {Vedral}}, \bibinfo {author} {\bibfnamefont {M.~B.}\ \bibnamefont {Plenio}}, \bibinfo {author} {\bibfnamefont {M.~A.}\ \bibnamefont {Rippin}},\ and\ \bibinfo {author} {\bibfnamefont {P.~L.}\ \bibnamefont {Knight}},\ }\bibfield  {title} {\bibinfo {title} {Quantifying entanglement},\ }\href {https://doi.org/10.1103/PhysRevLett.78.2275} {\bibfield  {journal} {\bibinfo  {journal} {Physical Review Letters}\ }\textbf {\bibinfo {volume} {78}},\ \bibinfo {pages} {2275} (\bibinfo {year} {1997})}\BibitemShut {NoStop}%
\bibitem [{\citenamefont {Cirac}\ \emph {et~al.}(2001)\citenamefont {Cirac}, \citenamefont {D{\"u}r}, \citenamefont {Kraus},\ and\ \citenamefont {Lewenstein}}]{cirac2001entangling}%
  \BibitemOpen
  \bibfield  {author} {\bibinfo {author} {\bibfnamefont {J.~I.}\ \bibnamefont {Cirac}}, \bibinfo {author} {\bibfnamefont {W.}~\bibnamefont {D{\"u}r}}, \bibinfo {author} {\bibfnamefont {B.}~\bibnamefont {Kraus}},\ and\ \bibinfo {author} {\bibfnamefont {M.}~\bibnamefont {Lewenstein}},\ }\bibfield  {title} {\bibinfo {title} {Entangling operations and their implementation using a small amount of entanglement},\ }\href {https://doi.org/10.1103/PhysRevLett.86.544} {\bibfield  {journal} {\bibinfo  {journal} {Physical Review Letters}\ }\textbf {\bibinfo {volume} {86}},\ \bibinfo {pages} {544} (\bibinfo {year} {2001})}\BibitemShut {NoStop}%
\bibitem [{\citenamefont {Bennett}\ \emph {et~al.}(1993)\citenamefont {Bennett}, \citenamefont {Brassard}, \citenamefont {Cr{\'e}peau}, \citenamefont {Jozsa}, \citenamefont {Peres},\ and\ \citenamefont {Wootters}}]{bennett1993teleporting}%
  \BibitemOpen
  \bibfield  {author} {\bibinfo {author} {\bibfnamefont {C.~H.}\ \bibnamefont {Bennett}}, \bibinfo {author} {\bibfnamefont {G.}~\bibnamefont {Brassard}}, \bibinfo {author} {\bibfnamefont {C.}~\bibnamefont {Cr{\'e}peau}}, \bibinfo {author} {\bibfnamefont {R.}~\bibnamefont {Jozsa}}, \bibinfo {author} {\bibfnamefont {A.}~\bibnamefont {Peres}},\ and\ \bibinfo {author} {\bibfnamefont {W.~K.}\ \bibnamefont {Wootters}},\ }\bibfield  {title} {\bibinfo {title} {Teleporting an unknown quantum state via dual classical and einstein-podolsky-rosen channels},\ }\href {https://doi.org/10.1103/PhysRevLett.70.1895} {\bibfield  {journal} {\bibinfo  {journal} {Physical review letters}\ }\textbf {\bibinfo {volume} {70}},\ \bibinfo {pages} {1895} (\bibinfo {year} {1993})}\BibitemShut {NoStop}%
\bibitem [{\citenamefont {Chitambar}\ \emph {et~al.}(2014)\citenamefont {Chitambar}, \citenamefont {Leung}, \citenamefont {Man{\v{c}}inska}, \citenamefont {Ozols},\ and\ \citenamefont {Winter}}]{chitambar2014everything}%
  \BibitemOpen
  \bibfield  {author} {\bibinfo {author} {\bibfnamefont {E.}~\bibnamefont {Chitambar}}, \bibinfo {author} {\bibfnamefont {D.}~\bibnamefont {Leung}}, \bibinfo {author} {\bibfnamefont {L.}~\bibnamefont {Man{\v{c}}inska}}, \bibinfo {author} {\bibfnamefont {M.}~\bibnamefont {Ozols}},\ and\ \bibinfo {author} {\bibfnamefont {A.}~\bibnamefont {Winter}},\ }\bibfield  {title} {\bibinfo {title} {Everything you always wanted to know about locc (but were afraid to ask)},\ }\href {https://doi.org/10.1007/s00220-014-1953-9} {\bibfield  {journal} {\bibinfo  {journal} {Communications in Mathematical Physics}\ }\textbf {\bibinfo {volume} {328}},\ \bibinfo {pages} {303} (\bibinfo {year} {2014})}\BibitemShut {NoStop}%
\bibitem [{\citenamefont {Emerson}\ \emph {et~al.}(2005)\citenamefont {Emerson}, \citenamefont {Alicki},\ and\ \citenamefont {{\.{Z}}yczkowski}}]{emerson2005scalable}%
  \BibitemOpen
  \bibfield  {author} {\bibinfo {author} {\bibfnamefont {J.}~\bibnamefont {Emerson}}, \bibinfo {author} {\bibfnamefont {R.}~\bibnamefont {Alicki}},\ and\ \bibinfo {author} {\bibfnamefont {K.}~\bibnamefont {{\.{Z}}yczkowski}},\ }\bibfield  {title} {\bibinfo {title} {Scalable noise estimation with random unitary operators},\ }\href {https://doi.org/10.1088/1464-4266/7/10/021} {\bibfield  {journal} {\bibinfo  {journal} {Journal of Optics B: Quantum and Semiclassical Optics}\ }\textbf {\bibinfo {volume} {7}},\ \bibinfo {pages} {S347} (\bibinfo {year} {2005})}\BibitemShut {NoStop}%
\bibitem [{\citenamefont {Knill}\ \emph {et~al.}(2008)\citenamefont {Knill}, \citenamefont {Leibfried}, \citenamefont {Reichle}, \citenamefont {Britton}, \citenamefont {Blakestad}, \citenamefont {Jost}, \citenamefont {Langer}, \citenamefont {Ozeri}, \citenamefont {Seidelin},\ and\ \citenamefont {Wineland}}]{knill2008randomized}%
  \BibitemOpen
  \bibfield  {author} {\bibinfo {author} {\bibfnamefont {E.}~\bibnamefont {Knill}}, \bibinfo {author} {\bibfnamefont {D.}~\bibnamefont {Leibfried}}, \bibinfo {author} {\bibfnamefont {R.}~\bibnamefont {Reichle}}, \bibinfo {author} {\bibfnamefont {J.}~\bibnamefont {Britton}}, \bibinfo {author} {\bibfnamefont {R.~B.}\ \bibnamefont {Blakestad}}, \bibinfo {author} {\bibfnamefont {J.~D.}\ \bibnamefont {Jost}}, \bibinfo {author} {\bibfnamefont {C.}~\bibnamefont {Langer}}, \bibinfo {author} {\bibfnamefont {R.}~\bibnamefont {Ozeri}}, \bibinfo {author} {\bibfnamefont {S.}~\bibnamefont {Seidelin}},\ and\ \bibinfo {author} {\bibfnamefont {D.~J.}\ \bibnamefont {Wineland}},\ }\bibfield  {title} {\bibinfo {title} {Randomized benchmarking of quantum gates},\ }\href {https://doi.org/10.1103/PhysRevA.77.012307} {\bibfield  {journal} {\bibinfo  {journal} {Physical Review A}\ }\textbf {\bibinfo {volume} {77}},\ \bibinfo {pages} {012307} (\bibinfo {year} {2008})}\BibitemShut {NoStop}%
\bibitem [{\citenamefont {Flammia}(2022)}]{flammia2021averaged}%
  \BibitemOpen
  \bibfield  {author} {\bibinfo {author} {\bibfnamefont {S.~T.}\ \bibnamefont {Flammia}},\ }\bibfield  {title} {\bibinfo {title} {{Averaged Circuit Eigenvalue Sampling}},\ }in\ \href {https://doi.org/10.4230/LIPIcs.TQC.2022.4} {\emph {\bibinfo {booktitle} {17th Conference on the Theory of Quantum Computation, Communication and Cryptography (TQC 2022)}}},\ \bibinfo {series} {Leibniz International Proceedings in Informatics (LIPIcs)}, Vol.\ \bibinfo {volume} {232},\ \bibinfo {editor} {edited by\ \bibinfo {editor} {\bibfnamefont {F.}~\bibnamefont {Le~Gall}}\ and\ \bibinfo {editor} {\bibfnamefont {T.}~\bibnamefont {Morimae}}}\ (\bibinfo  {publisher} {Schloss Dagstuhl -- Leibniz-Zentrum f{\"u}r Informatik},\ \bibinfo {address} {Dagstuhl, Germany},\ \bibinfo {year} {2022})\ pp.\ \bibinfo {pages} {4:1--4:10}\BibitemShut {NoStop}%
\bibitem [{\citenamefont {Chen}\ \emph {et~al.}(2021)\citenamefont {Chen}, \citenamefont {Yu}, \citenamefont {Zeng},\ and\ \citenamefont {Flammia}}]{chen2020robust}%
  \BibitemOpen
  \bibfield  {author} {\bibinfo {author} {\bibfnamefont {S.}~\bibnamefont {Chen}}, \bibinfo {author} {\bibfnamefont {W.}~\bibnamefont {Yu}}, \bibinfo {author} {\bibfnamefont {P.}~\bibnamefont {Zeng}},\ and\ \bibinfo {author} {\bibfnamefont {S.~T.}\ \bibnamefont {Flammia}},\ }\bibfield  {title} {\bibinfo {title} {Robust shadow estimation},\ }\href {https://doi.org/10.1103/PRXQuantum.2.030348} {\bibfield  {journal} {\bibinfo  {journal} {PRX Quantum}\ }\textbf {\bibinfo {volume} {2}},\ \bibinfo {pages} {030348} (\bibinfo {year} {2021})}\BibitemShut {NoStop}%
\bibitem [{\citenamefont {LeCam}(1973)}]{lecam1973convergence}%
  \BibitemOpen
  \bibfield  {author} {\bibinfo {author} {\bibfnamefont {L.}~\bibnamefont {LeCam}},\ }\bibfield  {title} {\bibinfo {title} {Convergence of estimates under dimensionality restrictions},\ }\href {https://doi.org/10.1214/aos/1193342380} {\bibfield  {journal} {\bibinfo  {journal} {The Annals of Statistics}\ ,\ \bibinfo {pages} {38}} (\bibinfo {year} {1973})}\BibitemShut {NoStop}%
\bibitem [{\citenamefont {{Blume-Kohout}}\ \emph {et~al.}(2013)\citenamefont {{Blume-Kohout}}, \citenamefont {{King Gamble}}, \citenamefont {{Nielsen}}, \citenamefont {{Mizrahi}}, \citenamefont {{Sterk}},\ and\ \citenamefont {{Maunz}}}]{Blume-Kohout2013Robust}%
  \BibitemOpen
  \bibfield  {author} {\bibinfo {author} {\bibfnamefont {R.}~\bibnamefont {{Blume-Kohout}}}, \bibinfo {author} {\bibfnamefont {J.}~\bibnamefont {{King Gamble}}}, \bibinfo {author} {\bibfnamefont {E.}~\bibnamefont {{Nielsen}}}, \bibinfo {author} {\bibfnamefont {J.}~\bibnamefont {{Mizrahi}}}, \bibinfo {author} {\bibfnamefont {J.~D.}\ \bibnamefont {{Sterk}}},\ and\ \bibinfo {author} {\bibfnamefont {P.}~\bibnamefont {{Maunz}}},\ }\href@noop {} {\bibinfo {title} {{Robust, self-consistent, closed-form tomography of quantum logic gates on a trapped ion qubit}}} (\bibinfo {year} {2013}),\ \Eprint {https://arxiv.org/abs/1310.4492} {arXiv:1310.4492 [quant-ph]} \BibitemShut {NoStop}%
\bibitem [{\citenamefont {Proctor}\ \emph {et~al.}(2017)\citenamefont {Proctor}, \citenamefont {Rudinger}, \citenamefont {Young}, \citenamefont {Sarovar},\ and\ \citenamefont {Blume-Kohout}}]{proctor2017randomized}%
  \BibitemOpen
  \bibfield  {author} {\bibinfo {author} {\bibfnamefont {T.}~\bibnamefont {Proctor}}, \bibinfo {author} {\bibfnamefont {K.}~\bibnamefont {Rudinger}}, \bibinfo {author} {\bibfnamefont {K.}~\bibnamefont {Young}}, \bibinfo {author} {\bibfnamefont {M.}~\bibnamefont {Sarovar}},\ and\ \bibinfo {author} {\bibfnamefont {R.}~\bibnamefont {Blume-Kohout}},\ }\bibfield  {title} {\bibinfo {title} {What randomized benchmarking actually measures},\ }\href {https://doi.org/10.1103/PhysRevLett.119.130502} {\bibfield  {journal} {\bibinfo  {journal} {Physical review letters}\ }\textbf {\bibinfo {volume} {119}},\ \bibinfo {pages} {130502} (\bibinfo {year} {2017})}\BibitemShut {NoStop}%
\bibitem [{\citenamefont {Nielsen}\ \emph {et~al.}(2021)\citenamefont {Nielsen}, \citenamefont {Gamble}, \citenamefont {Rudinger}, \citenamefont {Scholten}, \citenamefont {Young},\ and\ \citenamefont {Blume-Kohout}}]{nielsen2021gate}%
  \BibitemOpen
  \bibfield  {author} {\bibinfo {author} {\bibfnamefont {E.}~\bibnamefont {Nielsen}}, \bibinfo {author} {\bibfnamefont {J.~K.}\ \bibnamefont {Gamble}}, \bibinfo {author} {\bibfnamefont {K.}~\bibnamefont {Rudinger}}, \bibinfo {author} {\bibfnamefont {T.}~\bibnamefont {Scholten}}, \bibinfo {author} {\bibfnamefont {K.}~\bibnamefont {Young}},\ and\ \bibinfo {author} {\bibfnamefont {R.}~\bibnamefont {Blume-Kohout}},\ }\bibfield  {title} {\bibinfo {title} {Gate {S}et {T}omography},\ }\href {https://doi.org/10.22331/q-2021-10-05-557} {\bibfield  {journal} {\bibinfo  {journal} {{Quantum}}\ }\textbf {\bibinfo {volume} {5}},\ \bibinfo {pages} {557} (\bibinfo {year} {2021})}\BibitemShut {NoStop}%
\bibitem [{\citenamefont {{Huang}}\ \emph {et~al.}(2022)\citenamefont {{Huang}}, \citenamefont {{Flammia}},\ and\ \citenamefont {{Preskill}}}]{huang2022foundations}%
  \BibitemOpen
  \bibfield  {author} {\bibinfo {author} {\bibfnamefont {H.-Y.}\ \bibnamefont {{Huang}}}, \bibinfo {author} {\bibfnamefont {S.~T.}\ \bibnamefont {{Flammia}}},\ and\ \bibinfo {author} {\bibfnamefont {J.}~\bibnamefont {{Preskill}}},\ }\href@noop {} {\bibinfo {title} {{Foundations for learning from noisy quantum experiments}}} (\bibinfo {year} {2022}),\ \Eprint {https://arxiv.org/abs/2204.13691} {arXiv:2204.13691 [quant-ph]} \BibitemShut {NoStop}%
\bibitem [{\citenamefont {Chen}\ \emph {et~al.}(2023{\natexlab{b}})\citenamefont {Chen}, \citenamefont {Liu}, \citenamefont {Otten}, \citenamefont {Seif}, \citenamefont {Fefferman},\ and\ \citenamefont {Jiang}}]{chen2023learnability}%
  \BibitemOpen
  \bibfield  {author} {\bibinfo {author} {\bibfnamefont {S.}~\bibnamefont {Chen}}, \bibinfo {author} {\bibfnamefont {Y.}~\bibnamefont {Liu}}, \bibinfo {author} {\bibfnamefont {M.}~\bibnamefont {Otten}}, \bibinfo {author} {\bibfnamefont {A.}~\bibnamefont {Seif}}, \bibinfo {author} {\bibfnamefont {B.}~\bibnamefont {Fefferman}},\ and\ \bibinfo {author} {\bibfnamefont {L.}~\bibnamefont {Jiang}},\ }\bibfield  {title} {\bibinfo {title} {The learnability of pauli noise},\ }\href {https://doi.org/10.1038/s41467-022-35759-4} {\bibfield  {journal} {\bibinfo  {journal} {Nature Communications}\ }\textbf {\bibinfo {volume} {14}},\ \bibinfo {pages} {52} (\bibinfo {year} {2023}{\natexlab{b}})}\BibitemShut {NoStop}%
\bibitem [{\citenamefont {Chen}\ and\ \citenamefont {Gong}(2023)}]{chen2023futility}%
  \BibitemOpen
  \bibfield  {author} {\bibinfo {author} {\bibfnamefont {S.}~\bibnamefont {Chen}}\ and\ \bibinfo {author} {\bibfnamefont {W.}~\bibnamefont {Gong}},\ }\bibfield  {title} {\bibinfo {title} {Futility and utility of a few ancillas for pauli channel learning},\ }\href@noop {} {\bibfield  {journal} {\bibinfo  {journal} {arXiv preprint arXiv:2309.14326}\ } (\bibinfo {year} {2023})}\BibitemShut {NoStop}%
\bibitem [{\citenamefont {Choi}(1975)}]{choi1975completely}%
  \BibitemOpen
  \bibfield  {author} {\bibinfo {author} {\bibfnamefont {M.-D.}\ \bibnamefont {Choi}},\ }\bibfield  {title} {\bibinfo {title} {Completely positive linear maps on complex matrices},\ }\href {https://doi.org/10.1016/0024-3795(75)90075-0} {\bibfield  {journal} {\bibinfo  {journal} {Linear algebra and its applications}\ }\textbf {\bibinfo {volume} {10}},\ \bibinfo {pages} {285} (\bibinfo {year} {1975})}\BibitemShut {NoStop}%
\bibitem [{\citenamefont {Hoeffding}(1994)}]{hoeffding1994probability}%
  \BibitemOpen
  \bibfield  {author} {\bibinfo {author} {\bibfnamefont {W.}~\bibnamefont {Hoeffding}},\ }\bibfield  {title} {\bibinfo {title} {Probability inequalities for sums of bounded random variables},\ }\href@noop {} {\bibfield  {journal} {\bibinfo  {journal} {The collected works of Wassily Hoeffding}\ ,\ \bibinfo {pages} {409}} (\bibinfo {year} {1994})}\BibitemShut {NoStop}%
\end{thebibliography}%

\newpage
\widetext
\begin{center}
\textbf{\large Supplemental Materials: Tight bounds on Pauli channel learning without entanglement}
\end{center}
\setcounter{equation}{0}
\setcounter{figure}{0}
\setcounter{table}{0}
\setcounter{theorem}{0}
\setcounter{section}{0}
\makeatletter
\renewcommand{\thesection}{S\arabic{section}}
\renewcommand{\theequation}{S\arabic{equation}}
\renewcommand{\thefigure}{S\arabic{figure}}
\renewcommand{\bibnumfmt}[1]{[S#1]}
\renewcommand{\citenumfont}[1]{S#1}

\tableofcontents

\section{Preliminaries}\label{sec:pre}

For a Hilbert space $\mc H_S$ (sometimes denoted simply as $S$), we use $\mc L(H_S)$ to denote the set of linear operators acting on $\mc H_S$. 
A quantum channel on $S$ is a linear map $\mc L(\mc H_S)\mapsto\mc L(\mc H_S)$ that is completely-positive and trace-preserveing (CPTP)~\cite{nielsen2002quantum}. We will also make use of the completely-positive and trace-non-increasing map (CPTNI), which is defined to be a CP map that does not increase the trace of any input positive operator. A simple example of a CPTNI map is $\rho\mapsto K\rho K^\dagger$ for any $K$ satisfying $K^\dagger K\le I$.
A {quantum instrument} (or, a classical-quantum channel) on system $S$ and a classical register $A$ is a quantum channel of the following form,
\begin{equation}
    \mc C^{S\to AS}(\cdot) = \sum_j \ketbra{j}{j}^A\otimes\mc C_j^{S}(\cdot),\quad \mc C_j^S\in\mr{CPTNI}.
\end{equation}
Here $\{\ket{j}\}_j$ is an orthonormal basis representing the measurement outcome, and $\mc C_j$ controls both the measurement outcome probability distribution and the post-measurement states. We emphasis that the choice of $\mc C_j^S$ must guarantee $\mc C^{S\to AS}$ to be CPTP as a whole.

\medskip

We follow the notations of~\cite{chen2022quantum}. 
For an $n$-qubit Hilbert space, define ${\sf P}^n$ as the Pauli group modulo the global phase.
${\sf P}^n$ is an Abelian group isomorphic to $\mbb Z^{2n}_2$.
Specifically, we view every $a\in\mbb Z^{2n}_2$ as a $2n$-bit string $a={a_{x,1}a_{z,1}a_{x,2}a_{z,2}\cdots a_{x,n}a_{z,n}}$ corresponding to the Pauli operator
\begin{equation}    
	P_a = \otimes_{k=1}^n i^{a_{x,k}a_{z,k}} X^{a_{x,k}}Z^{a_{z,k}},
\end{equation}
where the phase is chosen to ensure Hermiticity.
We also define a sympletic inner product $\expval{\cdot,\cdot}$ within $\mbb Z_{2}^{2n}$ as 
\begin{equation}
	\expval{a,b} = \sum_{k=1}^n (a_{x,k}b_{z,k}+a_{z,k}b_{x,k}) \mod{2}.
\end{equation}
One can verify that $P_aP_b = (-1)^\expval{a,b}P_bP_a$~\cite{nielsen2002quantum}. We sometimes use $\mbb Z^{2n}_2$ and ${\sf P}^n$ interchangably, slightly abusing the notations.

\medskip
\noindent An $n$-qubit Pauli channel $\Lambda$ is a quantum channel of the following form
\begin{equation}
	\Lambda(\cdot) = \sum_{a\in\mbb Z_2^{2n}}p_a P_a(\cdot)P_a,
\end{equation}
where $\{p_a\}_a$ is called the \emph{Pauli error rates}. A linear map of the above form being CPTP is equivalent to that $\sum_a p_a = 1$ and $p_a\ge0$ for all $a$.

\medskip

\noindent An alternative expression for $\Lambda$ is
\begin{equation}
	\Lambda(\cdot) = \frac{1}{2^n}\sum_{b\in\mbb Z_2^{2n}}\lambda_b\Tr(P_b(\cdot))P_b,
\end{equation}
where $\{\lambda_b\}_b$ is called the \emph{Pauli eigenvalues}~\cite{flammia2020efficient,flammia2021pauli}. 
It is also known as the \emph{Pauli fidelities}, as $\lambda_a=\Tr[P_a\Lambda(P_a)]/2^n$.
For a CPTP map we necessarily have $\lambda_a\in[-1,1]$ for all $a$ and $\lambda_0=1$, but this is not sufficient as the completely-positivity is not guaranteed, so one needs to use extra caution when defining a Pauli channel with this representation.

These two sets of parameters are related by the Walsh-Hadamard transform
\begin{equation}\label{eq:walsh_hadamard}
	\begin{aligned}
		\lambda_b = \sum_{a\in\mbb Z_2^{2n}}p_a(-1)^\expval{a,b},\quad
		p_a = \frac{1}{4^n}\sum_{b\in\mbb Z_2^{2n}}\lambda_b(-1)^\expval{a,b}.
	\end{aligned}
\end{equation}

\medskip
\noindent We will also use the \emph{Pauli-transfer-matrix} (PTM) representation to simplify notations.
A linear operator $O$ acting on a $2^n$-dimensional Hilbert space can be viewed as a vector in a $4^n$-dimensional Hilbert space. We denote this vectorization of $O$ as $\lket{O}$ and the corresponding Hermitian conjugate as $\lbra{O}$. The inner product within this space is the \emph{Hilbert-Schmidt} product defined as $\lbraket{A}{B}\coleq \Tr(A^\dagger B)$. The \emph{normalized Pauli operators} $\{\sigma_a\coleq P_a/\sqrt{2^{n}},~a\in\mbb Z_2^{2n}\}$ forms an orthonormal basis for this space. 
In the PTM representation, a superoperator (\textit{i.e.}, quantum channel) becomes an operator acting on the $4^n$-dimensional Hilbert space, sometimes called the Pauli transfer operator. 
Explicitly, we have $\lket{\Lambda(\rho)} = \Lambda^\ptm\lket{\rho} \equiv \Lambda\lket{\rho}$,
where we use the same notation to denote a channel and its Pauli transfer operator, which should be clear from the context.
Specifically, a general Pauli channel $\Lambda$ has the following Pauli transfer operator
$$
\Lambda = \sum_{a\in\mbb Z_2^{2n}} \lambda_a\lketbra{\sigma_a}{\sigma_a}.
$$
        
\section{Models for learning without entanglement}\label{app:model}
    In this section, we formally define the separable schemes and the classical-memory-assisted schemes for quantum learning. 
    Then, we prove the equivalence between the two schemes. Namely, either one can simulate the other with the same sample complexity.

\begin{definition}
    A \textbf{separable scheme} (SEP) for learning properties from $N$ copies of a channel $\Lambda$ acting on $\mc L(\mc H_S)$ is specified by the following elements:
    \begin{enumerate}
        \item An arbitrarily large ancillary system $\mc H_A$ whose dimension can depend on $n$.
        \item A collection of processing channels $\{\mc C_t\}_{t=1}^{N-1}$ acting on $\mc L(\mc H_A\otimes\mc H_S)$ such that each $\mc C_t$ is separable across $A$ and $S$. That is, 
        $$
        \mc C_t = \sum_j \mc A^A_{t,j}\otimes\mc B^S_{t,j},\quad  ~\mr{where}~\mc A_{t,j},~\mc B_{t,j}~\mr{is~CPTNI}.
        $$
        \item An initial state $\rho_0$ and a final POVM measurement $\{E_k\}_k$ acting on $\mc H_{A}\otimes \mc H_{S}$ such that both of them are separable across $A$ and $S$. That is,
        $$
        \rho_0= \sum_j\sigma_j^A\otimes\gamma_j^S,\quad E_k=\sum_j M_{k,j}^A\otimes N_{k,j}^S,\quad~\mr{for}~\sigma_j,\gamma_j,M_{k,j},N_{k,j}\ge 0.
        $$
    \end{enumerate}
    The separable scheme works by inputing $\rho_0$, interleaving $N$ copies of $\Lambda$ with the processing channels $\mc C_t$, and measuring $\{E_k\}_k$ at the end, yielding the following outcome distribution:
    $$
    \Pr[k | \Lambda] = \Tr\left[E_k \left((\mathds 1^A\otimes\Lambda^S)\circ\mc C_{N-1}\circ\cdots\circ(\mathds 1^A\otimes\Lambda^S)\circ\mc C_1\circ(\mathds 1^A\otimes\Lambda^S)\right)(\rho_0)\right].
    $$
    Finally, an algorithm is specified to predict the desired properties based on the observed outcome $k$.
\end{definition}

\noindent Since we allow arbitrarily large ancilla for separable scheme, any adaptive control can be taken into account. We note again that separable operations contains LOCC as a subset.

\bigskip

\begin{definition}
    A \textbf{classical-memory-assisted scheme} for learning properties from $N$ copies of a channel $\Lambda$ acting on $\mc L(\mc H_S)$ is specified by the following elements:
    \begin{enumerate}
        \item $N$ arbitrarily large classical registers $\{A_t\}_{t=0}^{N-1}$.
        \item A collection of processing channels (or, quantum instruments) $\{\mc C^{\bm o_{<t}}\}_{t=1}^{N-1}$ defined as
        $$
        \mc C^{\bm o_{<t}}(\rho) \coleq \sum_{o_t} \ketbra{o_t}{o_t}^{A_{t}}\otimes\mc C_{o_t}^{\bm o_{<t}}(\rho)^S,\quad~\mr{for}~\mc C_{o_t}^{\bm o_{<t}}\in\mr{CPTNI}(\mc L(\mc H_S)),
        $$
        where the superscript of $o_{<t}$ indicates that the $t$-th processing channels can be chosen adaptively conditioned on the previous $t$ classical registers.
        \item An ensemble of initial states 
        $$\rho_0 \coleq \sum_{o_0}\ketbra{o_0}{o_0}^{A_0}\otimes\rho_{o_0}^S,$$ 
        where $\rho_{o_0}$ is some sub-normalized quantum state,
        and a final POVM measurement 
        $$\{E^{o_{<N}}_{o_N}\}^S_{o_N},$$ 
        which can be chosen adaptively conditioned on all previous outcomes.
    \end{enumerate}
    The scheme works by inputing the initial state, interleaving $N$ copies of $\Lambda$ with the processing channels, and measuring at the end, yielding the following outcome distribution:
    $$
    \Pr[\bm o | \Lambda] \equiv \Pr[o_{0:N} | \Lambda]= \Tr\left[E_{o_N}^{o_{<N}} \left(\Lambda\circ\mc C_{o_{N-1}}^{o_{<N-1}}\circ\cdots\circ\mc C_{o_2}^{o_{<2}}\circ\Lambda\circ\mc C_{o_1}^{o_0}\circ\Lambda\right)(\rho_{o_0})\right].
    $$
    Finally, an algorithm is specified to predict the desired properties based on the observed outcome $\bm o$.
\end{definition}

\noindent We remind that classical-memory-assisted schemes can be understood as quantum circuits assisted by mid-circuit measurement and feed-forward control.

\bigskip

We now prove the equivalence between these two schemes. We call a scheme $A$ can be \textbf{simulated} by another scheme $B$ if there exists a mapping $\mc S: \bm o_B\mapsto\bm o_A$ independent of $\Lambda$ such that $\Pr_A[\bm o_A |\Lambda] = \sum_{\bm o_B:\mc S(\bm o_B)=\bm o_A}\Pr_B[\bm o_B|\Lambda]$ for all $\bm o_A$ and $\Lambda$.
\begin{proposition}\label{prop:equivalence}
    Any separable scheme can be simulated by a classical-memory-assisted scheme using the same number of samples and vice versa.
\end{proposition}
\begin{proof}
    A classical-memory-assisted scheme can obviously be simulated by a separable scheme using the same number of samples, which is clear from Fig.~1 in the main text. Formally, choose the ancillary system to include all the classical registers, $\mc H_A = \otimes_{t=0}^{N-1}\mc H_{A_t}$, and choose the processing channel to be
    \begin{equation}
        \mc C_t(\cdot)\coleq \sum_{\bm o_{\le t}}
        \left(\bigotimes_{k=0}^{t-1}\ketbra{o_k}{o_k}(\cdot)\ketbra{o_k}{o_k}^{A_k}\right)\otimes \Tr_{A_t}(\cdot)\ketbra{o_t}{o_t}^{A_t}\otimes\mathds 1^{A_{t+1:N-1}}\otimes\mc C_{o_t}^{\bm o_{<t}}(\cdot)^S
    \end{equation}    
    which are separable channels between $A$ and $S$. In word, the $t$-th processing channel picks up the quantum instrument $\mc C^{\bm o_{<t}}$ according to $A_{0:t-1}$, writes down the outcome to $A_t$, and does nothing to $A_{t+1:N}$. Alternatively, this channel can be written in the PTM representation as
    \begin{equation}
        \mc C_t^{AS} = \sum_{\bm o_{\le t}}\left(\bigotimes_{k=0}^{t-1}\lketbra{o_k}{o_k}^{A_k}\right)\otimes\lketbra{o_t}{I}^{A_t}\otimes \mathds 1^{A_{t+1:N-1}}\otimes \mc C_{o_t}^{\bm o_{<t}}{}^S. 
    \end{equation}
    where $\lket{o_k}$ is defined to be the vectorization of the computational basis state $\ketbra{o_k}{o_k}$.
    
    \smallskip
    \noindent
    Additionally, choose the initial state as
    $$\sum_{o_0}\ketbra{o_0}{o_0}^{A_0}\otimes \rho_\mr{junk}^{A_{1:N-1}}\otimes\rho_{o_0}^S,$$
    where $\rho_\mr{junk}$ can be any quantum state,
    and the final POVM as 
    $$\left\{\left(\bigotimes_{t=0}^{N-1}\ketbra{o_t}{o_t}^{A_t}\right)\otimes\left(E_{o_N}^{o_{<N}}\right)^S\right\}_{\bm o}.$$ 
    One can verify the outcome distribution will be the same for both schemes.

    \medskip

    Now we prove the opposite direction. For a generic separable scheme, expand the outcome distribution,
    \begin{equation}
        \Pr[k|\Lambda] =\sum_{\bm j_{0:N}}\Tr\left[M_{k,j_N}(\mc A_{N-1,j_{N-1}}\cdots \mc A_{1,j_1})(\sigma_{j_0})\right]\Tr\left[N_{k,j_N}(\Lambda \mc B_{N-1,j_{N-1}}\cdots\Lambda\mc B_{1,j_1}\Lambda)(\gamma_{j_0})\right].
    \end{equation}
    A crucial observation is that, the first factor is independent of $\Lambda$ conditioned on $\bm j_{0:N}$. Therefore, instead of keeping track of the quantum state in the ancillary system, we just need to keep track of the index $\bm j$ using classical registers.
    Formally, we can design a classical-memory-assisted scheme with the following processing channels
    \begin{equation}
        \mc C^{j_{<t}}(\rho^S)\coleq \sum_{j_t}\ketbra{j_t}{j_t}^{A_t}\otimes\frac{\Tr\left[(\mc A_{t,j_{t}}\mc A_{t-1,j_{t-1}}\cdots \mc A_{1,j_1})(\sigma_{j_0})\right]}{\Tr\left[(\mc A_{t-1,j_{t-1}}\cdots \mc A_{1,j_1})(\sigma_{j_0})\right]}\mc B_{t,j_{t}}(\rho^S).
    \end{equation}
    One can verify that $\mc C^{j_{<t}}$ is CPTP. Additionally, choose the initial state to be 
    $$\rho_0 = \sum_{j_0}\ketbra{j_0}{j_0}^{A_0}\otimes\Tr[\sigma_{j_0}]\gamma_{j_0},$$ 
    and the final POVM measurement to be 
    \begin{equation}
        E_{k,j_{N}}^{j_{<N}} \coleq \frac{\Tr\left[M_{k,j_N}(\mc A_{N-1,j_{N-1}}\cdots \mc A_{1,j_1})(\sigma_{j_0})\right]}{\Tr\left[(\mc A_{N-1,j_{N-1}}\cdots \mc A_{1,j_1})(\sigma_{j_0})\right]}\cdot N_{k,j_N}.
    \end{equation}
    The outcome distribution can then be computed as
    \begin{equation}
        \Pr[k,j_{0:N}|\Lambda] =\Tr\left[M_{k,j_N}(\mc A_{N-1,j_{N-1}}\cdots \mc A_{1,j_1})(\sigma_{j_0})\right]\Tr\left[N_{k,j_N}(\Lambda \mc B_{N-1,j_{N-1}}\cdots\Lambda\mc B_{1,j_1}\Lambda)(\gamma_{j_0})\right].
    \end{equation}
    By tracing out the classical register $A_{0:N}$, we retrieve the distribution of the separable scheme. This means any separable scheme can be simulated by a classical-memory-assisted scheme using the same number of samples, which completes our proof. 
\end{proof}

\noindent \textbf{Sample complexity vs. Number of measurements.}\quad
Since the separable schemes and classical-memory-assisted schemes are equivalent in terms of sample complexity $N_\mr{samp}$, we will call both of them entanglement-free schemes. However, for the classical-memory-assisted schemes, there is an alternative figure of merit that is practically interesting, which we call \emph{the number of measurements}, defined as follows
\begin{definition}\label{def:Nmeas}
    The {number of measurements}, $N_\mr{meas}$, of a classical-memory-assisted scheme is defined to be the number of non-trivial quantum instruments used in the worst case (where the final POVM measurement also counts).
    Here, a quantum instrument 
    $\mc C^{S\to AS}=\sum_{t}\ketbra{t}{t}^A\otimes \mc C_{t}^{S}$
    is called trivial if every $\mc C_t^S$ is proportional to some $\mr{CPTP}$ map.
\end{definition}
\noindent The reason why we call such quantum instruments trivial is that they are not extracting any information from the input, hence not doing a measurement. Indeed, the probability of seeing $\ket t$ is independent of the input state for such quantum instruments.
We also emphasize that, for different measurement outcome sequence there can be different numbers of non-trivial quantum instruments, and our definition focus on the worst case.

By definition, $N_\mr{samp}\ge N_\mr{meas}$, as each $\Lambda$ is followed by one quantum instrument (or the final POVM measurement). The lower bound for the latter thus implies one for the former. In the following, when we talk about the number of measurements of an entanglement-free schemes, we are specifically referring to the classical-memory-assisted schemes.

\section{Tight bounds on Pauli channel learning}\label{app:bound}

In this section, we present our main technical results, stated in Theorem~\ref{th:main_SM}.

\begin{theorem}\label{th:main_SM}
    If there exists a classical-memory-assisted scheme that, for any $n$-qubit Pauli channel $\Lambda$, outputs an estimator $\widehat\lambda_a$ such that 
    $|\widehat\lambda_a-\lambda_a|\le\varepsilon\le1/6$ with probability at least $2/3$ for any $a\in{\sf P}^n$, after making $N$ rounds of measurement, then $N=\Omega(2^{n}/\varepsilon^2)$.
\end{theorem}

\begin{proof}
We first prove a lower bound for the sample complexity, and then strengthen it to the number of measurements.
Define the following set of Pauli channels
\begin{equation}
\begin{aligned}
    \Lambda_{a,\pm}&\coleq \lket{\sigma_0}\lbra{\sigma_0} \pm \varepsilon_{0}\lket{\sigma_a}\lbra{\sigma_a},\quad P_a\ne I;\\
    \Lambda_0&\coleq\lket{\sigma_0}\lbra{\sigma_0}.
\end{aligned}
\end{equation}
Here we will set $\varepsilon_0\le1/3$.
$\Lambda_0$ is known as the completely depolarizing channel. To see $\Lambda_{a,\pm}$ is CPTP, we just need to check the non-negativity of the associated Pauli error rates
\begin{equation}
    p_c = \frac{1}{4^n}\sum_b(-1)^\expval{c,b}\lambda_b = \frac{1}{4^n}(1\pm\varepsilon_0)\ge0,\quad \forall c\in {\sf P}^n.
\end{equation}
A learning algorithm satisfying Theorem~\ref{th:main_SM} with $\varepsilon=\varepsilon_0/2$ is able to distinguish any one of these Pauli channels with high success probability.
Now, we define a partially-revealed hypothesis-testing task similar to Ref.~\cite{huang2022quantum}:
A referee first samples an $a\in{\sf P}^n\backslash I$ and $s=\pm1$ according to uniform distribution. Then conducts one of the following actions with equal probability:
\begin{enumerate}
    \item Sending $N$ copies of $\Lambda_0$ to the player.
    \item Sending $N$ copies of $\Lambda_{a,s}$ to the player. 
\end{enumerate}
The player is then asked to measure the $N$ copies of channels with any schemes. After the measurement has been completed, the referee reveals the value of $a$ to the player. The player is now asked to guess which action the referee has taken, based on the obtained measurement outcome. Crucially, the player must completes all quantum measurement before the reveal of $a$, and can only do classical post-processing after that.

Suppose there exists a separable scheme satisfying the assumption of Theorem~\ref{th:main_SM}. Then the player can win the game with probability $2/3$ for any $(a,s)$: Just by querying $\lambda_{a}$ after the referee revealed the value of $a$, the player can distinguish among $\{\Lambda_{a,-},\Lambda_{0},\Lambda_{a,+}\}$ with 2/3 chance. 
Denote the measurement outcome distribution for $\Lambda_0$, $\Lambda_{a,\pm}$ as $p_0$, $p_{a,\pm}$, respectively.
Based on Le Cam's two-point method~\cite{lecam1973convergence}, we must have
\begin{equation}\label{eq:lower_tvd}
    \E_{a\ne 0}\mr{TVD}(p_0,~\mbb E_{s=\pm 1}[p_{a,s}]) \ge 
 (1-2*\frac13) = \frac{1}{3}.
\end{equation}
Now we compute the L.H.S., \textit{i.e.}, the average TVD. 
\begin{equation}
\E_{a\ne 0}\mr{TVD}(p_0,~\mbb E_{s=\pm 1}[p_{a,s}]) 
    = \mbb E_{a\ne 0}\sum_{\bm o}\max\left\{0, p_0(\bm o) - \mbb E_s p_{a,s}(\bm o)\right\}.
\end{equation}
We need some additional notations on the classical-memory-assisted scheme. Focusing on the $t$-th step, for every CP trace-non-increasing map $\mc C_{o_t}^{\bm o_{<t}}$, choosing any Kraus representation $\mc C_{o_t}^{\bm o_{<t}}(\cdot) = \sum_j K_j(\cdot)K_j^\dagger$~\cite{choi1975completely}, we define $E^{(\bm o_{<t})}_{o_t}\coleq \sum_jK_j^\dagger K_j$ as the associated POVM element. Indeed,
\begin{equation}\label{eq:effective_POVM}
    \Tr[E_{o_t}^{(\bm o_{<t})}\rho] = \Tr[\mc C_{o_t}^{\bm o_{<t}}(\rho)]. 
\end{equation}
On the other hand, we denote the state fed into the $t$-th $\Lambda_{a,s}$ as $\rho_{a,s}^{(\bm o_{<t})}$, which satisfies the following recurrence relation by definition
\begin{equation}
    \rho_{a,s}^{(\bm o_{<t+1})}\coleq \frac{\mc C_{o_t}^{\bm o_{<t}}\circ \Lambda(\rho_{a,s}^{(\bm o_{<t})})}{\Tr[\mc C_{o_t}^{\bm o_{<t}}\circ \Lambda(\rho_{a,s}^{(\bm o_{<t})})]},
\end{equation}
and the initial condition $\rho_{a,s}^{(\bm o_{<1})}\equiv\rho_\mr{init}$. Here we assume $\rho_\mr{init}$ to be fixed rather than drawing from an ensemble, without loss of generality, as the total variation distance is joint convex.

\medskip

Now, focus on the probability distribution difference,
\begin{align}
    &p_0(\bm o) - \mbb E_s p_{a,s}(\bm o) \\
    =~&\prod_{t=1}^N\lbra{E^{(\bm o_{<t})}_{o_t}}\Lambda_{0}\lket{\rho_{0}^{(\bm o_{<t})}} -  \mbb E_s \prod_{t=1}^N\lbra{E^{(\bm o_{<t})}_{o_t}}\Lambda_{a,s}\lket{\rho_{a,s}^{(\bm o_{<t})}}\\
    =~& \prod_{t=1}^N\frac{1}{2^n}\lbraket{E^{(\bm o_{<t})}_{o_t}}{I}\lbraket{I}{\rho_{0}^{(\bm o_{<t})}} -  \mbb E_s\prod_{t=1}^N\frac{1}{2^n}\left(\lbraket{E^{(\bm o_{<t})}_{o_t}}{I}\lbraket{I}{\rho_{a,s}^{(\bm o_{<t})}}+s\varepsilon_{0}\lbraket{E^{(\bm o_{<t})}_{o_t}}{P_a}\lbraket{P_a}{\rho_{a,s}^{(\bm o_{<t})}}\right)\\
    =~&\left(\prod_{t=1}^N \frac{1}{2^n}\Tr(E_{o_t}^{(\bm o_{<t})})\right)\left(1 - \mbb E_s\left(\prod_{t=1}^N\left(1 + s\varepsilon_{0}\frac{\Tr(E^{(\bm o_{<t})}_{o_t}{P_a})\Tr({P_a}{\rho_{a,s}^{(\bm o_{<t})})}}{\Tr(E_{o_t}^{(\bm o_{<t})})}\right)\right)\right)\\
    =~& p_0(\bm o)\left(1 - \mbb E_s\left(\prod_{t=1}^N\left(1 + s\varepsilon_{0}\frac{\Tr(E^{(\bm o_{<t})}_{o_t}{P_a})\Tr({P_a}{\rho_{a,s}^{(\bm o_{<t})})}}{\Tr(E_{o_t}^{(\bm o_{<t})})}\right)\right)\right)\label{eq:diff}
\end{align}
Here $\rho_{a,s}^{(\bm o_{<t})}$ is the post-measurement state after the ($t-1$)th measurement, conditioned on the outcome being $\bm o_{1:t-1}$.
Now we try to bound the following term,
\begin{align}
    &\mbb E_s\left(\prod_{t=1}^N\left(1 + s\varepsilon_{0}\frac{\Tr(E^{(\bm o_{<t})}_{o_t}{P_a})\Tr({P_a}{\rho_{a,s}^{(\bm o_{<t})})}}{\Tr(E_{o_t}^{(\bm o_{<t})})}\right)\right)\\
    =~& \E_s\exp\left(\sum_{t=1}^N\log\left(1 + s\varepsilon_{0}\frac{\Tr(E^{(\bm o_{<t})}_{o_t}{P_a})\Tr({P_a}{\rho_{a,s}^{(\bm o_{<t})})}}{\Tr(E_{o_t}^{(\bm o_{<t})})}\right)\right)\\
    \ge~&\exp\left(\frac12\sum_{t=1}^N\sum_{s=\pm1}\log\left(1 + s\varepsilon_{0}\frac{\Tr(E^{(\bm o_{<t})}_{o_t}{P_a})\Tr({P_a}{\rho_{a,s}^{(\bm o_{<t})})}}{\Tr(E_{o_t}^{(\bm o_{<t})})}\right)\right) \eqcol \bigstar,
\end{align}
where the second line uses the fact that each term in the product is non-negative, and the third line uses Jensen's inequality. Now we define
\begin{equation}
    \bar\rho_a^{(\bm o_{<t})} = \frac12(\rho_{a,+}^{(\bm o_{<t})} + \rho_{a,-}^{(\bm o_{<t})}),\quad 
    \Delta_a^{(\bm o_{<t})} = \frac12(\rho_{a,+}^{(\bm o_{<t})} - \rho_{a,-}^{(\bm o_{<t})}).
\end{equation}
Then above inequality can then be expressed as
\begin{align}
    \bigstar &= \exp\left(\frac12\sum_{t=1}^N\log\left[\left(1 + \varepsilon_{0}\frac{\Tr(E^{(\bm o_{<t})}_{o_t}{P_a})\Tr({P_a}{\rho_{a,+}^{(\bm o_{<t})})}}{\Tr(E_{o_t}^{(\bm o_{<t})})}\right)\left(1 - \varepsilon_{0}\frac{\Tr(E^{(\bm o_{<t})}_{o_t}{P_a})\Tr({P_a}{\rho_{a,-}^{(\bm o_{<t})})}}{\Tr(E_{o_t}^{(\bm o_{<t})})}\right)\right]\right) \\
    &= \exp\left(\frac12\sum_{t=1}^N\log\left[
    1 - \left( \varepsilon_{0}^2\frac{\Tr^2(E_{o_t}^{(\bm o_{<t})}P_a)}{\Tr^2(E_{o_t}^{(\bm o_{<t})})}\left(\Tr^2(P_a\bar\rho_{a}^{(\bm o_{<t})})-\Tr^2(P_a\Delta_{a}^{(\bm o_{<t})})\right) - 
    2\varepsilon_{0}\frac{\Tr(E^{(\bm o_{<t})}_{o_t}{P_a})\Tr({P_a}{\Delta_{a}^{(\bm o_{<t})})}}{\Tr(E_{o_t}^{(\bm o_{<t})})}\right)
    \right]\right)\\
    &\ge \exp\left(\frac12\sum_{t=1}^N\log\left[
    1 - \left( \varepsilon_{0}^2\frac{\Tr^2(E_{o_t}^{(\bm o_{<t})}P_a)}{\Tr^2(E_{o_t}^{(\bm o_{<t})})}\Tr^2(P_a\bar\rho_{a}^{(\bm o_{<t})}) + 2\varepsilon_{0}\left|\frac{\Tr(E^{(\bm o_{<t})}_{o_t}{P_a})\Tr({P_a}{\Delta_{a}^{(\bm o_{<t})})}}{\Tr(E_{o_t}^{(\bm o_{<t})})}\right|\right)
    \right]\right)\\
    &\ge\exp\left(
    - \sum_{t=1}^N\left( \varepsilon_{0}^2\frac{\Tr^2(E_{o_t}^{(\bm o_{<t})}P_a)}{\Tr^2(E_{o_t}^{(\bm o_{<t})})}\Tr^2(P_a\bar\rho_{a}^{(\bm o_{<t})})+ 2\varepsilon_{0}\left|\frac{\Tr(E^{(\bm o_{<t})}_{o_t}{P_a})\Tr({P_a}{\Delta_{a}^{(\bm o_{<t})})}}{\Tr(E_{o_t}^{(\bm o_{<t})})}\right|\right)
    \right)\\
    &\ge 1 - \sum_{t=1}^N\left( \varepsilon_{0}^2\frac{\Tr^2(E_{o_t}^{(\bm o_{<t})}P_a)}{\Tr^2(E_{o_t}^{(\bm o_{<t})})}\Tr^2(P_a\bar\rho_{a}^{(\bm o_{<t})}) + 2\varepsilon_{0}\left|\frac{\Tr(E^{(\bm o_{<t})}_{o_t}{P_a})\Tr({P_a}{\Delta_{a}^{(\bm o_{<t})})}}{\Tr(E_{o_t}^{(\bm o_{<t})})}\right|\right).
\end{align}
The fourth line uses the fact that $\log(1-x)\ge -2x$ for $x\in[0,0.79]$ and that the terms inside the inner braket is upper bounded by $\varepsilon_{0}^2 +2\varepsilon_{0} \le 0.78$. The last line uses $\exp(-x)\ge1-x$ for any $x$. Now substitute this back to the expression of avergae TVD.
\begin{align}
    &\E_{a\ne 0}\mr{TVD}(p_0,~\mbb E_{s=\pm 1}[p_{a,s}]) \\
    \le~& \E_{a\ne 0}\sum_{\bm o}\max\left\{0, p_0(\bm o)\sum_{t=1}^N\left( \varepsilon_{0}^2\frac{\Tr^2(E_{o_t}^{(\bm o_{<t})}P_a)}{\Tr^2(E_{o_t}^{(\bm o_{<t})})}\Tr^2(P_a\bar\rho_{a}^{(\bm o_{<t})}) + 2\varepsilon_{0}\left|\frac{\Tr(E^{(\bm o_{<t})}_{o_t}{P_a})\Tr({P_a}{\Delta_{a}^{(\bm o_{<t})})}}{\Tr(E_{o_t}^{(\bm o_{<t})})}\right|\right)\right\}\\
    =~&\mbb E_{a\ne 0}\sum_{\bm o} p_0(\bm o)\sum_{t=1}^N\left( \varepsilon_{0}^2\frac{\Tr^2(E_{o_t}^{(\bm o_{<t})}P_a)}{\Tr^2(E_{o_t}^{(\bm o_{<t})})}\Tr^2(P_a\bar\rho_{a}^{(\bm o_{<t})}) + 2\varepsilon_{0}\left|\frac{\Tr(E^{(\bm o_{<t})}_{o_t}{P_a})\Tr({P_a}{\Delta_{a}^{(\bm o_{<t})})}}{\Tr(E_{o_t}^{(\bm o_{<t})})}\right|\right).\label{eq:improvable}
\end{align}
The last line is because the second argument of $\max$ is now nonnegative. 
Now we bound the two terms from above. For the first term, note that
\begin{align}
    \mbb E_{a\ne 0}\frac{\Tr^2(E_{o_t}^{(\bm o_{<t})}P_a)}{\Tr^2(E_{o_t}^{(\bm o_{<t})})}\Tr^2(P_a\bar\rho_{a}^{(\bm o_{<t})}) &\le \mbb E_{a\ne 0}\frac{\Tr^2(E_{o_t}^{(\bm o_{<t})}P_a)}{\Tr^2(E_{o_t}^{(\bm o_{<t})})}\\
    &=\frac{1}{4^n-1}\frac{\sum_a\Tr^2(E_{o_t}^{(\bm o_{<t})}P_a)-\Tr^2(E_{o_t}^{(\bm o_{<t})})}{\Tr^2(E_{o_t}^{(\bm o_{<t})})}\\
    &=\frac{1}{4^n-1}\frac{2^n\Tr({E_{o_t}^{(\bm o_{<t})}}^2)-\Tr^2(E_{o_t}^{(\bm o_{<t})})}{\Tr^2(E_{o_t}^{(\bm o_{<t})})}\\
    &\le \frac{2^n}{4^n-1}.
\end{align}
The third line is Pauli twirling. The last line uses the fact that $E_{o_t}^{(\bm o_{<t})}$ is positive semi-definite and that $l_2$-norm is upper bounded by $l_1$-norm.

\medskip

Here we pause to remark that, if channel concatenation were not allowed, the $t$-th input state $\rho_{a,s}^{\bm o_{<t}}$ would be independent of $(a,s)$ conditioned on $\bm o_{<t}$, and thus $\Delta_a^{(\bm o_{<t})} = 0$. 
In that case, we would only have the first term, and the proof would already be completed, similar to \cite{huang2022quantum,chen2022exponential}.
The major challenge in our setting is how to address the second term.

\medskip
\noindent For the second term, note that
\begin{align}
    \E_{a\ne 0}\left|\frac{\Tr(E^{(\bm o_{<t})}_{o_t}{P_a})\Tr({P_a}{\Delta_{a}^{(\bm o_{<t})})}}{\Tr(E_{o_t}^{(\bm o_{<t})})}\right|
    &\le \sqrt{\E_{a\ne0}\frac{\Tr^2(E_{o_t}^{(\bm o_{<t})}P_a)}{\Tr^2(E_{o_t}^{(\bm o_{<t})})}}\sqrt{\E_{a\ne0}\Tr^2(P_a\Delta_a^{(\bm o_{<t})})}\label{eq:cauchy-second}\\
    &\le \sqrt{\frac{2^n}{4^n-1}}\sqrt{\E_{a\ne0}\Tr^2(P_a\Delta_a^{(\bm o_{<t})})}.
\end{align}
The first line is Cauchy-Schwarz. The second line is the same as the derivation for the first term.
The remaining of the proof is to upper bound the following
\begin{align}
\E_{a\ne 0}\Tr^2(P_a\Delta_a^{(\bm o_{<t})}) &= \E_{a\ne 0}\Tr^2(P_a(\rho_{a,+}^{(\bm o_{<t})} - \rho^{(\bm o_{<t})}_{a,-})/2)
= \E_{a\ne 0}\left(\left(\mu_{a,+}^{(\bm o_{<t})}-\mu_{a,-}^{(\bm o_{<t})}\right)/2\right)^2.
\end{align}
where we've defined $\mu_{a,\pm}^{(\bm o_{<t})}\coleq \Tr(P_a\rho_{a,\pm}^{(\bm o_{<t})})$. We have the following recurrence relation
\begin{equation}
    \rho_{a,\pm}^{(\bm o_{<t+1})} = \frac{(\mc C_{o_t}^{(\bm o_{<t})}\circ\Lambda_{a,\pm})(\rho_{a,\pm}^{(\bm o_{<t})})}{\Tr\left[(\mc C_{o_t}^{(\bm o_{<t})}\circ\Lambda_{a,\pm})(\rho_{a,\pm}^{(\bm o_{<t})})\right]} = 
    \frac{\mc C_{o_t}^{(\bm o_{<t})}(I\pm\varepsilon_{0} \mu_{a,\pm}^{(\bm o_{<t})}P_a)}{\Tr\left[\mc C_{o_t}^{(\bm o_{<t})}(I\pm\varepsilon_{0} \mu_{a,\pm}^{(\bm o_{<t})}P_a)\right]},
\end{equation}
Now, let $c^{(\bm o_{\le t})}_{a,b}\coleq \Tr[P_a \mc C_{o_t}^{(\bm o_{<t})}(P_b)]/2^n$, which is the Pauli transfer matrix representation of $\mc C_{o_t}^{(\bm o_{<t})}$. Taking expectation value of $P_a$ on both side of the above equation, we get
\begin{equation}\label{eq:recurrence}
    \mu_{a,\pm}^{(\bm o_{<t+1})} =\frac{c_{a,0}^{(\bm o_{\le t})} \pm \varepsilon_{0}\mu_{a,\pm}^{(\bm o_{<t})}c_{a,a}^{(\bm o_{\le t})}}{c_{0,0}^{(\bm o_{\le t})} \pm \varepsilon_{0}\mu_{a,\pm}^{(\bm o_{<t})}c_{0,a}^{(\bm o_{\le t})}}.
\end{equation}

For clarity, we omit the superscript of $c_{a,b}$ and denote $\mu^{(\bm o_{<t})}_{a,\pm}\equiv\mu_{a,\pm}$, $\mu^{(\bm o_{<t+1})}_{a,\pm}\equiv\mu'_{a,\pm}$.
The difference between $\mu_{a,\pm}'$ satisfies
\begin{align}
    &\mu_{a,+}' - \mu_{a,-}'\\ 
    =~& \left(\frac{c_{a,0} }{c_{0,0} + \varepsilon_{0}\mu_{a,+}c_{0,a}} - \frac{c_{a,0} }{c_{0,0} - \varepsilon_{0}\mu_{a,-}c_{0,a}}\right) +\varepsilon_{0}\left(\frac{\mu_{a,+}c_{a,a} }{c_{0,0} + \varepsilon_{0}\mu_{a,+}c_{0,a}} + \frac{\mu_{a,-}c_{a,a} }{c_{0,0} - \varepsilon_{0}\mu_{a,-}c_{0,a}}\right)\\
    \equiv~& (X) + (Y),
\end{align}
where we denote the first and second braket as $X,Y$, respectively. Now we have that
\begin{align}
    |X| &\le \frac{|c_{a,0}|}{(1-\varepsilon_{0})c_{0,0}}-\frac{|c_{a,0}|}{(1+\varepsilon_{0})c_{0,0}} = \frac{2\varepsilon_{0}}{1-\varepsilon_{0}^2}\frac{|c_{a,0}|}{c_{0,0}}.
\end{align}
as $|\mu_{a,\pm}|\le 1$ and $|c_{0,a}|\le c_{0,0}$. To see the later, notice that every completely-positive map has a Kraus representation~\cite{choi1975completely}, thus
\begin{align}
    |c_{a,b}| &\equiv \left|\Tr[P_a\mc C_{o_t}(P_b/2^n)]\right| \\
    &\le\sum_j\left|\Tr[P_aK_jP_bK_j^\dagger]/2^n\right|\\
    &\le \sum_j \sqrt{\Tr[P_aK_jK_j^\dagger P_a^\dagger]}\sqrt{\Tr[P_bK_j^\dagger K_j P_b^\dagger]}/2^n\\
    &= \sum_j\Tr[K_jK_j^\dagger]/2^n = \Tr[\mc C_{o_t}(I/2^n)]\equiv c_{0,0},
\end{align}
where the third line is Cauchy-Schwarz. For the second term,
\begin{equation}
    |Y|\le \frac{\varepsilon_{0}(|\mu_{a,+}|+|\mu_{a,-}|)|c_{a,a}|}{(1-\varepsilon_{0})c_{0,0}}\le \frac{\varepsilon_{0}}{1-\varepsilon_{0}}(|\mu_{a,+}|+|\mu_{a,-}|),
\end{equation}
as $|c_{a,a}|\le c_{0,0}$ which has been proved.
Therefore,
\begin{align}
    \E_{a\ne 0}\Tr^2(P_a\Delta_a^{(\bm o_{<t})}) 
    &= \E_{a\ne0}((X+Y)/2)^2\\
    &\le \left(\E_{a\ne0}X^2 + \E_{a\ne0}Y^2\right)/2\\
    &\le \frac12(\frac{2\varepsilon_{0}}{1-\varepsilon_{0}^2})^2\E_{a\ne0}\frac{c_{a,0}^2}{c_{0,0}^2} + \frac12(\frac{\varepsilon_{0}}{1-\varepsilon_{0}})^2\E_{a\ne0}(|\mu_{a,+}|+|\mu_{a,-}|)^2\\
    &\le \frac12(\frac{2\varepsilon_{0}}{1-\varepsilon_{0}^2})^2\frac{2^n}{4^n-1} + \frac12(\frac{\varepsilon_{0}}{1-\varepsilon_{0}})^2\E_{a\ne0}(2\mu^2_{a,+}+2\mu^2_{a,-}).\label{eq:recurrence_Delta}
\end{align}
We use $(a+b)^2\le 2a^2 +2b^2$ in the second and last line. For the last line we also use 
\begin{align}
    \E_{a\ne0}\frac{c_{a,0}^2}{c_{0,0}^2} \equiv \frac{\E_{a\ne0}\Tr^2(P_a\mc C_{o_t}(I))}{\Tr^2(\mc C_{o_t}(I))}\le \frac{2^n}{4^n-1}\frac{\Tr\left((\mc C_{o_t}(I))^2\right)}{\Tr^2(\mc C_{o_t}(I))} \le \frac{2^n}{4^n-1},
\end{align}
The first and second inequalities are Pauli twirling and monotonicity of $l_p$ norm, respectively.
Now we just need to bound $\E_{a\ne0}\mu^2_{a,\pm}.$ Again, by the recurrence relation Eq.~\eqref{eq:recurrence}, we have
\begin{align}
    \E_{a\ne0}\left(\mu'_{a,\pm}\right)^2 &=
    \E_{a\ne0}\left(\frac{c_{a,0} \pm \varepsilon_{0}\mu_{a,\pm} c_{a,a}}{c_{0,0} \pm \varepsilon_{0}\mu_{a,\pm}c_{0,a}}\right)^2\\
    &\le \frac{\E_{a\ne 0}\left(c_{a,0} \pm \varepsilon_{0}\mu_{a,\pm} c_{a,a}\right)^2}{(1-\varepsilon_{0})^2c_{0,0}^2}\\
    &\le\frac{2\E_{a\ne 0}c_{a,0}^2 + 2\varepsilon_{0}^2\E_{a\ne 0}\mu_{a,\pm}^2c_{a,a}^2}{(1-\varepsilon_{0})^2c_{0,0}^2}\\
    &\le \frac{2}{(1-\varepsilon_{0})^2}\frac{2^n}{4^n-1} + \frac{2\varepsilon_{0}^2}{(1-\varepsilon_{0})^2}\E_{a\ne 0}\mu_{a,\pm}^2.
\end{align}
We also have the initial condition (note that $\rho_\mr{init}$ has no dependence on $a$)
\begin{equation}
    \E_{a\ne0}\Tr^2[P_a\rho_\mr{init}] \le \frac{2^n}{4^n-1} < \frac{2}{(1-\varepsilon_{0})^2}\frac{2^n}{4^n-1}.
\end{equation}
By induction, we have the following bound,
\begin{align}
    \E_{a\ne 0}\Tr^2[P_a\rho_{a,\pm}^{(\bm o_{<t})}] \equiv \E_{a\ne 0}(\mu_{a,\pm}^{(\bm o_{<t})})^2 &\le \frac{1-\left(\frac{2\varepsilon_{0}^2}{(1-\varepsilon_{0})^2}\right)^{t-1}}{1-\frac{2\varepsilon_{0}^2}{(1-\varepsilon_{0})^2}} \frac{2\cdot 2^{n}}{(1-\varepsilon_{0})^2(4^n-1)} \\
    &\le \frac{2}{1-2\varepsilon_{0}-\varepsilon_{0}^2}\cdot\frac{2^{n}}{4^n-1},
\end{align}
which holds since $2\varepsilon_{0}^2/(1-\varepsilon_{0})^2\le1$ for $\varepsilon_{0}\le1/3$.
Substitute this back to Eq.~\eqref{eq:recurrence_Delta},
\begin{align}
    \E_{a\ne0}\Tr^2(P_a\Delta_a^{(\bm o_{<t})}) &\le \frac12(\frac{2\varepsilon_{0}}{1-\varepsilon_{0}^2})^2\frac{2^n}{4^n-1} + \frac12(\frac{\varepsilon_{0}}{1-\varepsilon_{0}})^2\frac{8}{1-2\varepsilon_{0}-\varepsilon_{0}^2}\cdot\frac{2^n}{4^n-1}\\
    &= \frac12\left((\frac{2}{1-\varepsilon_{0}^2})^2+\frac{1}{(1-\varepsilon_{0})^2}\frac{8}{1-2\varepsilon_{0}-\varepsilon_{0}^2}\right)\cdot\frac{\varepsilon_{0}^2 2^n}{4^n-1}\\
    &\eqcol f(\varepsilon_{0})\cdot\frac{\varepsilon_{0}^2 2^n}{4^n-1},
\end{align}
where $f(\varepsilon_{0})=O(1)$ for $\varepsilon_{0}\le1/3$. 
Substitute this back to Eq.~\eqref{eq:cauchy-second} and Eq.~\eqref{eq:improvable}, we got the following TVD bound
\begin{equation}\label{eq:upper_tvd}
\begin{aligned}
    \E_{a\ne 0}\mr{TVD}(p_0,\E_{s=\pm1}[p_{a,s}]) &\le    \sum_{\bm o} p_0(\bm o)\sum_{t=1}^N\left(\varepsilon_{0}^2\frac{2^n}{4^n-1} + 2\varepsilon_{0}^2\frac{2^n}{4^n-1}\sqrt{f(\varepsilon_{0})}\right)\\
    &= N\left(\varepsilon_{0}^2\frac{2^n}{4^n-1} + 2\varepsilon_{0}^2\frac{2^n}{4^n-1}\sqrt{f(\varepsilon_{0})}\right),
\end{aligned}
\end{equation}
which, combined with Eq.~\eqref{eq:lower_tvd}, gives the following sample complexity lower bound
\begin{equation}\label{eq:bound_with_f}
    N\ge \frac{1}{3(1+2\sqrt{f(\varepsilon_{0})})}\frac{4^n-1}{2^n\varepsilon_{0}^2} =
    \frac{1}{12(1+2\sqrt{f(2\varepsilon)})}\frac{4^n-1}{2^n\varepsilon^2}.
\end{equation}
Finally, note that $f(2\varepsilon)<44$ for $\varepsilon\in[0,1/6]$, thus we have
\begin{equation}
    N\ge 0.005 \cdot \frac{4^n-1}{2^n\varepsilon^2}= \Omega(2^n/\varepsilon^2).
\end{equation}
Now we explain how this can be strengthened to a lower bound for $N_\mr{meas}$ (i.e., the number of measurements).
Suppose the quantum instrument $\mc C^{\bm o_{<t}}$ is trivial in the sense of Definition~\ref{def:Nmeas}. Then, the effective POVM elements $E_{o_t}^{\bm o_{<t}}$ for all $o_t$ will be proportional to $I$, according to the discussion above Eq.~\eqref{eq:effective_POVM}.
Now, for any fixed $\bm o$ in Eq.~\eqref{eq:improvable}, the $t$-th term in the sum will contribute $0$ if $E_{o_t}^{\bm o_{<t}}$ is proportional to $I$, as $\Tr(I\cdot P_a)=0$ for all $a\ne0$.
Therefore, there are at most $N_\mr{meas}$ non-zero terms in the sum, and we have a uniform upper bound for each of them. The upper bound on TVD, Eq.~\eqref{eq:upper_tvd}, can thus be strengthened to 
\begin{equation}
        \E_{a\ne 0}\mr{TVD}(p_0,\E_{s=\pm1}[p_{a,s}]) \le N_\mr{meas}\left(\varepsilon_{0}^2\frac{2^n}{4^n-1} + 2\varepsilon_{0}^2\frac{2^n}{4^n-1}\sqrt{f(\varepsilon_{0})}\right)
\end{equation}
Therefore, the same lower bound for $N$ holds for $N_\mr{meas}$. This completes the proof of Theorem~\ref{th:main_SM}.
\end{proof}

\section{Bounds on learning identifiable Pauli noise}\label{app:degeneracy}
\begin{definition}\label{def:coarse}
    Given a partition of ${\sf P}^n\backslash I$ into $\bm B\coleq\{B_1,\cdots,B_K\}$. For any $n$-qubit Pauli channel $\Lambda$, the \textbf{geometrically-averaged Pauli fidelity} according to $\bm B$ are defined to be
    $$
    \lambda_{B}\coleq \mr{sgn}\left(\prod_{b\in B}\lambda_b\right)\cdot\left|\prod_{b\in B} \lambda_b\right|^{\cfrac{1}{|B|}},\quad \forall ~B\in \bm B.
    $$
    Here $\mr{sgn}(\cdot)$ is the sign function.
\end{definition}

\begin{theorem}\label{th:coarse_SM}
    For any partition $\bm B$ in Def.~\ref{def:coarse} such that $C \coleq \max_{B\in\bm B}|B|$, 
    if there exists an classical-memory-assisted scheme that, for any $n$-qubit Pauli channel, outputs an estimator $\widehat\lambda_{B}$ such that $|\widehat\lambda_{B}-\lambda_{B}|\le\varepsilon\le1/6C$ with probability at least $2/3$ for any $B\in\bm B$,
    after making $N$ rounds of measurements, then $N=\Omega({2^n}\varepsilon^{-2}{C^{-2}})$.
\end{theorem}

\begin{proof}
    The proof is similar to that of Theorem~\ref{th:main_SM}. Here, we need to define a different set of Pauli channels, as follows
    \begin{align}
        \Lambda_{B,\pm}&\coleq \lketbra{\sigma_0}{\sigma_0}\pm\varepsilon_{0}\E_{b\in B}\lketbra{\sigma_b}{\sigma_b} \\&= \lketbra{\sigma_0}{\sigma_0}\pm\frac{\varepsilon_{0}}{|B|}\sum_{b\in B}\lketbra{\sigma_b}{\sigma_b},\quad \forall~B\in\bm B.\\
        \Lambda_{0}&\coleq \lketbra{\sigma_0}{\sigma_0}.
    \end{align}
    Here we need $\varepsilon_0\le1/3$. To see $\Lambda_{B,\pm}$ is indeed CPTP, check the associated Pauli error rates:
    \begin{align}
        p_a = \frac{1}{4^n}\sum_{b}(-1)^\expval{a,b}\lambda_b
        =\frac{1}{4^n}\left(1+\sum_{b\in B}(-1)^\expval{a,b}\lambda_b\right)
        \ge \frac{1}{4^n}\left( 1 - |B|\cdot\frac{\varepsilon_0}{|B|}\right)\ge0.
    \end{align}
    Furthermore, define the following distribution over $\bm B$: 
    \begin{equation}
        \pi(B)={|B|}/({4^n-1}).
    \end{equation} 
    Now we are ready to define the partially-revealed hypothesis testing task: The referee samples $B$ according to $\pi$ and $s=\pm1$ with equal probability. Then, sends either $N$ copies of $\Lambda_0$ or $N$ copies of $\Lambda_{B,s}$ to the player. After the player completes their measurement, the referee reveals $B$. The player wins if they can guess which action has been taken. Similarly as before, if there exists a scheme satisfying the assumption of Theorem~\ref{th:coarse_SM} with $\varepsilon \coleq \varepsilon_0/(2C)$, the player wins with high probability by querying $\lambda_B$, and we have the following TVD lower bounds
    \begin{equation}
        \E_{B\sim\pi}\mr{TVD}(p_0(\bm o),\E_{s}p_{B,s}(\bm o)) \ge 1/3.
    \end{equation}
    Now, calculate the probability distribution difference,
    \begin{align}
        p_0(\bm o) - \E_s p_{B,s}(\bm o) &= p_0(\bm o)\left(1-\frac{\E_sp_{B,s}(\bm o)}{p_0(\bm o)}\right)\\
        &=p_0(\bm o)\left(1-\frac{\E_sp_{B,s}(\bm o)}{p_0(\bm o)}\right)\\
        &=p_0(\bm o)\left(1-\E_s\prod_{t=1}^N\left(1+s\varepsilon_{0}
        \E_{b\in B}\frac{\Tr(E_{o_t}^{(\bm o_{<t})}P_b)\Tr(P_b\rho_{B,s}^{(\bm o_{<t})})}{\Tr(E_{o_t}^{(\bm o_{<t})})}\right)\right).
    \end{align}
    Define
    \begin{equation}
    \bar\rho_B^{(\bm o_{<t})} = \frac12(\rho_{B,+}^{(\bm o_{<t})} + \rho_{B,-}^{(\bm o_{<t})}),\quad 
    \Delta_B^{(\bm o_{<t})} = \frac12(\rho_{B,+}^{(\bm o_{<t})} - \rho_{B,-}^{(\bm o_{<t})}).
    \end{equation}
    Focus on lower bounding the following term:
    \begin{align}
    &\E_s\prod_{t=1}^N\left(1+s\varepsilon_{0}
    \E_{b\in B}\frac{\Tr(E_{o_t}^{(\bm o_{<t})}P_b)\Tr(P_b\rho_{B,s}^{(\bm o_{<t})})}{\Tr(E_{o_t}^{(\bm o_{<t})})}\right)\\
    \ge~&\exp\left(\frac12\sum_{t=1}^N\sum_{s=\pm1}\log\left(1+s\varepsilon_{0}
    \E_{b\in B}\frac{\Tr(E_{o_t}^{(\bm o_{<t})}P_b)\Tr(P_b\rho_{B,s}^{(\bm o_{<t})})}{\Tr(E_{o_t}^{(\bm o_{<t})})}\right)\right)\\
    \ge~&\exp\left(\frac12\sum_{t=1}^N\log\left(1-\varepsilon_{0}^2
    \left(\E_{b\in B}\frac{\Tr(E_{o_t}^{(\bm o_{<t})}P_b)\Tr(P_b\bar{\rho}_{B}^{(\bm o_{<t})})}{\Tr(E_{o_t}^{(\bm o_{<t})})}\right)^2
    -2\varepsilon_{0}\left|\E_{b\in B}\frac{\Tr(E_{o_t}^{(\bm o_{<t})}P_b)\Tr(P_b\Delta_{B}^{(\bm o_{<t})})}{\Tr(E_{o_t}^{(\bm o_{<t})})}\right|\right)\right)\\
    \ge~&\exp\left(\frac12\sum_{t=1}^N\log\left(1-\E_{b\in B}\left(\varepsilon_{0}^2
    \frac{\Tr^2(E_{o_t}^{(\bm o_{<t})}P_b)\Tr^2(P_b\bar{\rho}_{B}^{(\bm o_{<t})})}{\Tr^2(E_{o_t}^{(\bm o_{<t})})}+2\varepsilon_{0}\left|\frac{\Tr(E_{o_t}^{(\bm o_{<t})}P_b)\Tr(P_b\Delta_{B}^{(\bm o_{<t})})}{\Tr(E_{o_t}^{(\bm o_{<t})})}\right|\right)\right)\right)\\
    \ge~& 1 - \E_{b\in B}\sum_{t=1}^N\left(\varepsilon_{0}^2
    \frac{\Tr^2(E_{o_t}^{(\bm o_{<t})}P_b)\Tr^2(P_b\bar{\rho}_{B}^{(\bm o_{<t})})}{\Tr^2(E_{o_t}^{(\bm o_{<t})})}+2\varepsilon_{0}\left|\frac{\Tr(E_{o_t}^{(\bm o_{<t})}P_b)\Tr(P_b\Delta_{B}^{(\bm o_{<t})})}{\Tr(E_{o_t}^{(\bm o_{<t})})}\right|\right).
    \end{align}
    This gives us the following average TVD bound
    \begin{align}
        &\E_{B\sim \pi}\mr{TVD}(p_0(\bm o), \mbb E_s p_{B,s}(\bm o)) \\
        \le~& \sum_{\bm o}p_0(\bm o)\sum_{t=1}^N\E_{B\sim\pi}\E_{b\in B} \left(\varepsilon_{0}^2
        \frac{\Tr^2(E_{o_t}^{(\bm o_{<t})}P_b)\Tr^2(P_b\bar{\rho}_{B}^{(\bm o_{<t})})}{\Tr^2(E_{o_t}^{(\bm o_{<t})})}+2\varepsilon_{0}\left|\frac{\Tr(E_{o_t}^{(\bm o_{<t})}P_b)\Tr(P_b\Delta_{B}^{(\bm o_{<t})})}{\Tr(E_{o_t}^{(\bm o_{<t})})}\right|\right).
    \end{align}
    For the first term, note that
    \begin{align}
        \E_{B\sim\pi}\E_{b\in B}\frac{\Tr^2(E_{o_t}^{(\bm o_{<t})}P_b)\Tr^2(P_b\bar{\rho}_{B}^{(\bm o_{<t})})}{\Tr^2(E_{o_t}^{(\bm o_{<t})})}&\le \E_{B\sim\pi}\E_{b\in B}\frac{\Tr^2(E_{o_t}^{(\bm o_{<t})}P_b)}{\Tr^2(E_{o_t}^{(\bm o_{<t})})}\\
        &= \sum_{B\in\bm B}\frac{|B|}{4^n-1}\sum_{b\in B}\frac{1}{|B|} \frac{\Tr^2(E_{o_t}^{(\bm o_{<t})}P_b)}{\Tr^2(E_{o_t}^{(\bm o_{<t})})}\\
        &= \E_{a\ne 0}\frac{\Tr^2(E_{o_t}^{(\bm o_{<t})}P_a)}{\Tr^2(E_{o_t}^{(\bm o_{<t})})}\\
        &\le \frac{2^n}{4^n-1}.
    \end{align}
    For the second term,
    \begin{align}
        &\E_{B\sim\pi}\E_{b\in B}\left|\frac{\Tr^2(E_{o_t}^{(\bm o_{<t})}P_b)\Tr^2(P_b\Delta_{B}^{(\bm o_{<t})})}{\Tr^2(E_{o_t}^{(\bm o_{<t})})}\right| \\
        \le~& \sqrt{\E_{B\sim\pi}\E_{b\in B}\frac{\Tr^2(E_{o_t}^{(\bm o_{<t})}P_b)}{\Tr^2(E_{o_t}^{(\bm o_{<t})})}}\sqrt{\E_{B\sim\pi}\E_{b\in B}\Tr^2(P_b\Delta_B^{(\bm o_{<t})})}\\
        \le~&\sqrt{\frac{2^n}{4^n-1}}\sqrt{\E_{B\sim\pi}\E_{b\in B}\Tr^2(P_b\Delta_B^{(\bm o_{<t})})}
    \end{align}
    Now we bound the expression inside the second square root. Define $\mu_{b,\pm}^{(\bm o_{<t})}\coleq \Tr(P_b\rho_{B,\pm}^{(\bm o_{<t})})$ and $c^{(\bm o_{\le t})}_{a,b}\coleq\Tr[P_a\mc C_{o_t}^{(\bm o_{<t})}(P_b)]/2^n$. Note that the $B$ dependence is implicit in $\mu_{b,\pm}$ since every $b$ belongs to one and only one $B$. First notice the following recurrence relation for $\mu_{b,\pm}$,
    \begin{equation}
        \mu_{b,\pm}^{(\bm o_{<t+1})}\coleq \frac{c^{(\bm o_{\le t})}_{b,0}\pm \varepsilon_{0}\E_{b'\in B}\mu_{b',\pm}^{(\bm o_{<t})}c^{(\bm o_{\le t})}_{b,b'}}{c^{(\bm o_{\le t})}_{0,0}\pm \varepsilon_{0}\E_{b'\in B}\mu_{b',\pm}^{(\bm o_{<t})}c^{(\bm o_{\le t})}_{0,b'}} 
    \end{equation}
    To simplify notation, for a fixed $0\le t\le N-1$, we use $\mu,\mu'$ to denote $\mu^{(\bm o_{<t})},\mu^{(\bm o_{<t+1})}$, respectively, and drop the superscript for $c$. We first study the difference between plus and minus term
    \begin{align}
        &\left|\mu'_{b,+} - \mu'_{b,-}\right|\\
        \le~&\left|\frac{c_{b,0}}{c_{0,0}+\varepsilon_{0}\E_{b'\in B}\mu_{b',+}c_{0,b'}}-\frac{c_{b,0}}{c_{0,0}-\varepsilon_{0}\E_{b'\in B}\mu_{b',-}c_{0,b'}}\right| \\&\quad\quad\quad+ 
        \varepsilon_{0}\left|\frac{\E_{b'\in B}\mu_{b',+}c_{b,b'}}{c_{0,0}+\varepsilon_{0}\E_{b'\in B}\mu_{b',+}c_{0,b'}}+\frac{\E_{b'\in B}\mu_{b',-}c_{b,b'}}{c_{0,0}-\varepsilon_{0}\E_{b'\in B}\mu_{b',-}c_{0,b'}}\right|\\
        \le~& \frac{2\varepsilon_{0}}{1-\varepsilon_{0}^2}\frac{|c_{b,0}|}{c_{0,0}} + \frac{\varepsilon_{0}}{1-\varepsilon_{0}}\E_{b'\in B} (|\mu_{b',+}|+|\mu_{b',-}|).
    \end{align}
    Here we use $|c_{a,b}|\le c_{0,0}$ as proved in the last section. Therefore,
    \begin{align}
        &\E_{B\sim\pi}\E_{b\in B}\Tr^2(P_b\Delta_B^{(\bm o_{<t})}) \\
        \equiv~& \E_{B\sim\pi}\E_{b\in B}\left(\frac{\mu'_{b,+}-\mu'_{b,-}}2\right)^2\\
        \le~& \frac12\left(\frac{2\varepsilon_{0}}{1-\varepsilon_{0}^2}\right)^2\E_{B\sim\pi}\E_{b\in B}\frac{c_{b,0}^2}{c_{0,0}^2} + \frac12\left(\frac{\varepsilon_{0}}{1-\varepsilon_{0}}\right)^2\E_{B\sim\pi}\left(\E_{b'\in B} (|\mu_{b',+}|+|\mu_{b',-}|)\right)^2\\
        \le~& \frac12\left(\frac{2\varepsilon_{0}}{1-\varepsilon_{0}^2}\right)^2\frac{2^n}{4^n-1} + \left(\frac{\varepsilon_{0}}{1-\varepsilon_{0}}\right)^2 \E_{B\sim\pi}\E_{b'\in B}(\mu_{b',+}^2+\mu_{b',-}^2)\label{eq:Delta_exp_2}.
    \end{align}
    Note that the second term in the third line has no $b$ dependence, so we omit $\E_{b\in B}$. To bound the second term in the last line, again by the recurrence relation
    \begin{align}
        \E_{B\sim\pi}\E_{b\in B}(\mu'_{b,\pm})^2 &= \E_{B\sim\pi}\E_{b\in B}\frac{(c_{b,0}\pm\varepsilon_{0}\E_{b'\in B}\mu_{b',\pm}c_{b,b'})^2}{(c_{0,0}\pm\varepsilon_{0}\E_{b'\in B}\mu_{b',\pm}c_{0,b'})^2}\\
        &\le \E_{B\sim\pi}\E_{b\in B}\frac{(c_{b,0}\pm\varepsilon_{0}\E_{b'\in B}\mu_{b',\pm}c_{b,b'})^2}{(1-\varepsilon_{0})^2c_{0,0}^2}\\
        &\le \E_{B\sim\pi}\E_{b\in B}\frac{2c^2_{b,0}+2\varepsilon_{0}^2(\E_{b'\in B}\mu_{b',\pm}c_{b,b'})^2}{(1-\varepsilon_{0})^2c_{0,0}^2}\\
        &\le \frac{2}{(1-\varepsilon_{0})^2}\frac{2^n}{4^n-1} + \frac{2\varepsilon_{0}^2}{(1-\varepsilon_{0})^2}\E_{B\sim\pi}\E_{b'\in B}\mu_{b',\pm}^2.
    \end{align}
    The initial condition is
    \begin{align}
        \E_{B\sim\pi}\E_{b'\in B}\mu_{b',\pm}^{(\mr{init})^2}
        \equiv\E_{B\sim\pi}\E_{b'\in B}\Tr^2(P_b\rho_0) \le \frac{2^n}{4^n-1}<\frac{2}{(1-\varepsilon_{0})^2}\frac{2^n}{4^n-1}.
    \end{align}
    By induction, as long as $2\varepsilon_0^2/(1-\varepsilon_0)^2\le1$,
    \begin{equation}
        \E_{B\sim\pi}\E_{b'\in B}\mu_{b',\pm}^2 \le \frac{1}{1-\frac{2\varepsilon_{0}^2}{(1-\varepsilon_{0})^2}}\frac{2}{(1-\varepsilon_{0})^2}\frac{2^n}{4^n-1}\le\frac{2}{1-2\varepsilon_{0}-\varepsilon_{0}^2}\frac{2^n}{4^n-1}.
    \end{equation}
    Substitute this into Eq.~\eqref{eq:Delta_exp_2},
    \begin{align}
        \E_{B\sim\pi}\E_{b'\in B}\Tr^2(P_b\Delta_{B}^{(\bm o_{<t})}) &\le \frac12\left(\frac{2}{1-\varepsilon_{0}^2}\right)^2\frac{2^n\varepsilon_{0}^2}{4^n-1} + \left(\frac{1}{1-\varepsilon_{0}}\right)^2 \frac{4}{1-2\varepsilon_{0}-\varepsilon_{0}^2}\frac{2^n\varepsilon_{0}^2}{4^n-1}\\
        &\eqcol f(\varepsilon_{0})\frac{2^n\varepsilon_{0}^2}{4^n-1}. 
    \end{align}
    Here $f(\varepsilon_{0})<44$ for $\varepsilon_{0}\le 1/3$. Putting this back to the TVD bound,
    \begin{equation}
        \E_{B\sim\pi}\mr{TVD}(p_0(\bm o),\E_{s}p_{B,s}(\bm o))\le N\left( (1+2\sqrt{f(\varepsilon_{0})})\frac{2^n\varepsilon_{0}^2}{4^n-1}\right).
    \end{equation}
    Thus,
    \begin{equation}
        N \ge \frac{1}{3(1+2\sqrt{f(\varepsilon_{0})})}\frac{4^n-1}{2^n\varepsilon_{0}^2} = \frac{1}{12(1+2\sqrt{f(2C\varepsilon)})}\frac{4^n-1}{2^n\varepsilon^2C^2} \ge 0.005 \cdot\frac{4^n-1}{2^n\varepsilon^2C^2}= \Omega(2^n\varepsilon^{-2}C^{-2}).
    \end{equation}
    Based on the same argument as in the previous section, the same lower bounds hold for $N_\mr{meas}$. This completes the proof of Theorem~\ref{th:coarse_SM}.
\end{proof}

\section{Numerical comparison of upper and lower bounds}\label{app:numeric}
In this section, we provide more details of the comparison given in main text Fig.~2. We will first review the entanglement-assisted Pauli channel learning protocol proposed in~\cite{chen2022quantum} which assumed perfect Bell state preparation and Bell measurement. For practical consideration, we will adapt the protocol to noisy Bell state/measurement and derive the sample complexity upper bound. Finally, we will recall the exact form of the ancilla-free lower bound from~\cite{chen2022quantum}, which is to be compared with our improved lower bound from Theorem~\ref{th:main_SM}.

Given an $n$-qubit main system $S$ and an $n$-qubit ancillary system $A$. The Bell states on $SA$ are defined as
\begin{equation}
    \ket{\Psi_a}\coleq P_a^S\otimes I^A \ket{\Psi_0}, \quad\forall a\in{\sf P}^{n},
\end{equation}
where $\ket{\Psi_0} = \ket{\psi_+}^{\otimes n} \equiv \left(\frac1{\sqrt{2}}\left(\ket{00}+\ket{11}\right)\right)^{\otimes n}$ is the $n$-fold tensor product a $2$-qubit Bell pair between the system and ancilla. Use the fact that $\ketbra{\psi_+}{\psi_+} = \frac{1}{4}\sum_{a\in {\sf P}^1}P_a\otimes P_a^{\sf T}$, the density matrix of an $n$-qubit Bell state can be expressed as
\begin{equation}
\begin{aligned}
        \ketbra{\Psi_a}{\Psi_a} &= (P_a\otimes I) \ketbra{\Psi_0}{\Psi_0}(P_a\otimes I)\\ 
        &= \frac{1}{4^n}\sum_{b\in{\sf P}^n}P_aP_bP_a\otimes P_b^{\sf T}\\
        &= \frac{1}{4^n}\sum_{b\in{\sf P}^n}(-1)^{\expval{a,b}}P_b\otimes P_b^{\sf T}.
\end{aligned}
\end{equation}
Here $\expval{a,b}=\mathds{1}[P_a~\text{commutes with}~P_b]$ where $\mathds 1$ is the indicator function. One can verify that all the $4^n$ Bell states form an orthogonal basis, the projective measurement onto which is known as the Bell measurement. The basic step of entanglement-assisted Pauli channel learning protocol studied in~\cite{chen2022quantum} includes preparing $\ket{\Psi_0}$, applying $\Lambda$ on the main system, and measuring in the Bell basis $\{\ketbra{\Psi_a}{\Psi_a}\}_a$. When the Bell states/measurement are perfect, the outcome distribution of this protocol is exactly the Pauli error rate of $\Lambda$. To analyze the effect of noisy Bell states/measurement, we assume i.i.d. single-qubit depolarizing noise happens on each qubit during the protocol, \textit{i.e.},
\begin{equation}
    \mc E_p(\rho) = (1-p)\rho + p\Tr[\rho]I/2.
\end{equation}
The fidelity of a single Bell pair going through the depolarizing noises is
\begin{equation}\label{eq:fid_vs_p}
\begin{aligned}
    F_\mr{Bell} &= \expval{\psi_+|\mc E_p^{\otimes 2}(\ketbra{\psi_+}{\psi_+})|\psi_+} = \frac{1+3(1-p)^2}{4}.
\end{aligned}
\end{equation}
With this noise model, the measurement outcome probability can then be calculated as
\begin{equation}
    \begin{aligned}
    \widetilde p(a)&= \expval{\Psi_a|(\Lambda\otimes\mathcal I)\circ\mc E_p^{\otimes 2n}\left(\ketbra{\Psi_0}{\Psi_0}\right)|\Psi_a}\\
    &= \expval{\Psi_a|(\Lambda\otimes\mathcal I)\left(
    \frac{1}{4^n}\sum_{b\in{\sf P}^n}(1-p)^{2|b|}P_b\otimes P_b^{\sf T}
    \right)|\Psi_a}\\
    &= \expval{\Psi_a|\left(
    \frac{1}{4^n}\sum_{b\in{\sf P}^n}(1-p)^{2|b|}\lambda_b P_b\otimes P_b^{\sf T}
    \right)|\Psi_a}\\
    &= \frac{1}{4^{2n}}\sum_{b,c\in{\sf P}^n}(-1)^\expval{a,c}(1-p)^{2|b|}\lambda_b\Tr[P_bP_c]\Tr[P_b^{\sf T}P_c^{\sf T}]\\
    &= \frac{1}{4^n}\sum_{b\in{\sf P}^n}(-1)^\expval{a,b}(1-p)^{2|b|}\lambda_b,
    \end{aligned}
\end{equation}
where $|a|$ denotes the Pauli weight of $P_a$ (\textit{i.e.}, the number of non-identity single-qubit Pauli in $P_a$). If we set our estimator to be $\hat\lambda_b[a]\coleq (1-p)^{-2|b|}(-1)^\expval{a,b}$, then
\begin{equation}
    \mbb E_{a\sim {\tilde p}}\hat\lambda_b[a] = (1-p)^{-2|b|}\sum_{a\in{\sf P}^n}(-1)^\expval{a,b}\widetilde p(a) = \lambda_b,
\end{equation}
using the Walsh-Hadamard transform as in Eq.~\eqref{eq:walsh_hadamard}. Thus $\hat\lambda_b[a]$ is an unbiased estimator for $\lambda_b$. Plus, $|\hat\lambda_b[a]|\le(1-p)^{-2|b|}$. According to the Hoeffding's inequality~\cite{hoeffding1994probability}, by independently sampling 
\begin{equation}\label{eq:noisy_bell_upper}
    N = 2\varepsilon^{-2}(1-p)^{-4|b|}\log(2/\delta)
\end{equation}
measurement outcomes $\{a^{(k)}\}_{k=1}^N$ from $\widetilde{p}$ and averaging the estimators $\hat\lambda_b\coleq\hat\lambda_b[a^{(1)},\cdots,a^{(N)}]=\frac1N\sum_{k=1}^N\hat\lambda_b[a^{(k)}]$, we have $|\hat\lambda_b-\lambda_b|\le \varepsilon$ with probability at least $1-\delta$. In Fig.~2, we plot the upper bound of Eq.~\eqref{eq:noisy_bell_upper} by choosing $|b|=n$, which are the most difficult Pauli eigenvalues to learn, $\delta=1/3,~\varepsilon=0.1$, and choosing different $p$ depending on the desired Bell pair fidelity according to Eq.~\eqref{eq:fid_vs_p}.

\medskip

The best previously known ancilla-free lower bound for this task is~\cite[Theorem 6]{chen2022quantum} which allows channel concatenation and adaptive control but not mid-circuit measurement. The theorem says that, to learn all $\lambda_b$ to an additive precision $1/2$ with success probability at least $2/3$, the following number of measurements are necessary
\begin{equation}\label{eq:aflower}
    N \ge \frac{1}{6}(2^n-1)^{1/3}.
\end{equation}
Note that, this bound fixes the precision parameter to be $\varepsilon=1/2$, and it is not obvious how to generalize the method there to obtain an $\varepsilon$-dependent bound. Since a lower bound for one precision ($\varepsilon=1/2$) is also trivially a lower bound for a better precision (say, $\varepsilon=0.1$), we use Eq.~\eqref{eq:aflower} in Fig. 2.

\medskip

Finally, for our improved lower bound, we plot
\begin{equation}
    N \ge 0.01 \frac{4^n-1}{2^n\varepsilon^2},
\end{equation}
which can be obtained from Eq.~\eqref{eq:bound_with_f} given that $\varepsilon=0.1$. 

\end{document}